\documentclass[11pt]{article}
\usepackage{acronym}
\usepackage{amsmath}
\usepackage{amssymb}
\usepackage[american]{babel}
\usepackage{amsthm}

\usepackage{graphicx}
\usepackage[pdfborder={0 0 0}]{hyperref} 

\usepackage{epsfig}
\usepackage{dsfont}
\usepackage{gensymb}
\usepackage{hyperref}
\usepackage{fullpage}

\usepackage[svgnames]{xcolor}
\usepackage{tikz}
\usetikzlibrary{decorations.markings}
\usetikzlibrary{shapes.geometric}
\usetikzlibrary{arrows}
\pgfdeclarelayer{edgelayer}
\pgfdeclarelayer{nodelayer}
\pgfsetlayers{edgelayer,nodelayer,main}

\tikzstyle{none}=[inner sep=0pt]
\definecolor{green}{rgb}{0.000,1.000,0.000}
\definecolor{black}{rgb}{0.000,0.000,0.000}
\definecolor{red}{rgb}{1.000,0.000,0.012}
\definecolor{white}{rgb}{1.000,1.000,1.000}
\definecolor{lightgray}{rgb}{0.8,0.8,0.8}
\definecolor{darkgray}{rgb}{0.4,0.4,0.4}
\definecolor{gray}{rgb}{0.6,0.6,0.6}

\tikzstyle{simple}=[-,draw=black,line width=0.75]
\tikzstyle{dash}=[-,dashed,draw=black,line width=0.75]
\tikzstyle{graydash}=[-,dashed,draw=gray,line width=0.75]
\tikzstyle{wide}=[-,draw=black,line width=2.0]

\tikzstyle{dasharrow}=[->,dashed,draw=black, line width=0.75]

\tikzstyle{vertex}=[circle,fill=black,draw=black, scale = 0.3]
\tikzstyle{bigvertex}=[circle,fill=black,draw=black, scale = 0.5]

\pdfoutput=1

\usepackage{subfig}
\usepackage{booktabs}
\newcommand{\mycirc}[1][black]{\Large\textcolor{#1}{\ensuremath\bullet}}
\definecolor{goldenyellow}{rgb}{1.0, 0.87, 0.0}
\definecolor{limegreen}{rgb}{0.2, 0.8, 0.2}
\definecolor{pansypurple}{rgb}{0.47, 0.09, 0.29}

{\end{list}}

{\end{list}}

{\end{list}}

{\end{list}}

{\end{list}}

{\end{list}}

{\end{list}}

{\end{list}}

{\end{list}}

{\end{list}}

{\end{list}}

\newenvironment{dense_itemize}{%
\begin{list}{$\bullet$}%
{\setlength{\topsep}{1mm}%
\setlength{\partopsep}{0mm}%
\setlength{\parskip}{0mm}%
\setlength{\parsep}{0mm}%
\setlength{\itemsep}{0mm}%
\setlength{\labelwidth}{4mm}%
\setlength{\leftmargin}{0mm}%
\addtolength{\leftmargin}{\labelwidth}%
\addtolength{\leftmargin}{\labelsep}%
\setlength{\itemindent}{0mm}}}%
{\end{list}}

\graphicspath{{./graphics/}}

\newcommand{\todo}[1]
{
    \marginpar
        {\bf {[!!!]}}
        {\bf {#1}}
}

\def\comic#1#2#3{\parbox{#1}{\centering\includegraphics[width=#1]{#2}\\{\footnotesize #3}}}
\def\comicII#1#2{\parbox{#1}{\centering\includegraphics[width=#1]{#2}}}

\newtheorem{observation}{Observation}
\newtheorem{lemma}{Lemma}
\newtheorem{theorem}{Theorem}

\newtheorem{corollary}{Corollary}

\title{Combinatorics and complexity of guarding polygons with edge and point 2-transmitters\thanks{Abstracts of part of this work appeared in the informal workshops FWCG~\cite{cfils-npppe-14} and 
EuroCG~\cite{cfils-ce2ta-15}}}

\author{Sarah Cannon\thanks{{Supported by a Clare Boothe Luce Graduate Fellowship and NSF DGE-1148903.}}\\
College of Computing, Georgia Institute of Technology, Atlanta, USA.\\
 Email: {\tt sarah.cannon@gatech.edu}
\and 
Thomas G. Fai\\
School of Engineering and Applied Sciences, Harvard University, MA, USA. \\
Email: {\tt tfai@seas.harvard.edu}
\and 
Justin Iwerks\\
Mathematics Department, The Spence School, NY, USA. \\
Email: {\tt jiwerks@gmail.com}
\and
Undine Leopold\\
Mathematics Department, TU Chemnitz, Germany. \\
Email: {\tt undine.leopold@mathematik.tu-chemnitz.de}
\and
Christiane Schmidt\thanks{Supported by the Israeli Centers of Research Excellence (I-CORE) program (Center No. 4/11).}\\
The Rachel and Selim Benin School of Computer Science and Engineering,\\ The Hebrew University of Jerusalem. \\
Email: {\tt cschmidt@cs.huji.ac.il}
}

\begin{document}
\date{}
\maketitle

\abstract{
We consider a generalization of the classical Art Gallery Problem, where instead of a light source, the guards, called $k$-{\it transmitters}, model a wireless device with a signal that can pass through at most $k$ walls. We show it is NP-hard to compute a minimum cover of point 2-transmitters, point $k$-transmitters, and edge 2-transmitters in a simple polygon. The point 2-transmitter result extends to orthogonal polygons.
In addition, we give necessity and sufficiency results for the number of edge 2-transmitters in general, monotone, orthogonal monotone, and orthogonal polygons.

}

\section{Introduction}

The traditional art gallery problem (AGP) considers placing guards in an art gallery---modeled by a polygon---so that every point in the room can be seen by some guard. A similar question asks how to place wireless routers so that an entire room has a strong 
signal.  
Observation shows that often not only the distance from a modem, but also the number of walls a signal has to pass through, influences signal strength. 

Aichholzer et al.~\cite{affhhuv-mimp-09} first formalized this problem by considering {\it $k$-modems} 
({\it $k$-transmitters}), devices whose wireless 
signal can pass through at most $k$ walls. 
Since 2010, little progress has been made on the problem of $k$-transmitters, or even the problem of 2-transmitters,  despite reaching a wide audience as the topic of a computational geometry column by Joseph O'Rourke~\cite{o-cgc52-12} in the SIGACT News in 2012.
Analogous to the original 
AGP ($k=0$), two main questions can be considered:
\begin{itemize}
\item[(1)] Given a polygon $P$, can a minimum cardinality $k$-transmitter cover be computed efficiently?
\item[(2)] Given a class of polygons of $n$ vertices, what are lower and upper bounds on the number of guards needed to cover a polygon from this class? 
\end{itemize}

For the classical AGP, 
the complexity question (1) was answered with NP-hardness for many variants.
O'Rourke and Supowit~\cite{os-snpdp-83} gave a reduction from 3SAT, for polygons with holes and guards
restricted to lie on vertices. 
Lee
and Lin~\cite{ll-ccagp-86} 
gave the result for simple polygons.
This result was extended to point guards (that are allowed to be located anywhere inside of $P$) by Aggarwal
(see~\cite{o-agta-87}); Schuchardt and Hecker~\cite{sh-tnhag-95} gave NP-hardness proofs for rectilinear simple
polygons, both for point and vertex guards. 
The complexity of the $k$-/ 2-transmitter problem had not previously been settled, and in this paper, we prove the minimum point 2-transmitter, the minimum point $k$-transmitter, and the minimum edge 2-transmitter problems to be NP-hard in simple polygons. The minimum point 2-transmitter result also holds for simple, orthogonal polygons.

Answers to (2) are often referred to as ``Art Gallery theorems'', e.g. Chv\'atal's tight bound of $\lfloor \frac{n}{3}\rfloor$ for simple polygons~\cite{c-ctpg-75}.
Fisk~\cite{f-spcwt-78} later gave a short and simple proof for Chv\'atal's result. 
In the case of orthogonal polygons, the bound becomes $\lfloor \frac{n}{4}\rfloor$, as shown by Kahn et al.~\cite{kkk-tgrfw-83}.

For $k$-transmitters, Aichholzer 
et al.~\cite{affhhuv-mimp-09}~showed 
$\lceil \frac{n}{2k}\rceil$ $k$-transmitters are always sufficient and $\lceil \frac{n}{2k+4}\rceil$ $k$-transmitters are sometimes necessary to cover a monotone $n$-gon
\footnote{The stated lower bound of $\lceil n/(2k+2) \rceil$ given in \cite{affhhuv-mimp-09} is a typo, and the example only necessitates $\lceil n/(2k+4) \rceil$ 2-transmitters.}; for monotone orthogonal polygons they gave a tight bound of $\lceil \frac{n-2}{2k+4}\rceil$ $k$-transmitters, for $k$ even and $k=1$.
Aichholzer et al.~\cite{aff-mimp-14} improved the bounds on monotone polygons to a tight value of $\lceil \frac{n-2}{2k+3}\rceil$.
In addition, they gave tight bounds for  monotone orthogonal polygons for all values of $k$.
Other publications explored $k$-transmitter coverage of regions other than simple polygons, such as coverage of the plane in the presence of line or line segment obstacles \cite{bbal-cktpo-10,mvu-mip-09}. For example, Ballinger et al. established that for disjoint segments in the plane, where each segment has one of two slopes and the entire plane is to be covered, $\lceil \frac{1}{2}(\frac{5}{6}^{\log(k+1)}n+1)\rceil$ $k$-transmitters are always sufficient, and $\lceil \frac{n+1}{2k+2}\rceil$ $k$-transmitters are sometimes necessary. For polygons, Ballinger et al. concentrated on a class of spiral polygons, so called {\it spirangles}, and established that $\lfloor\frac{n}{8}\rfloor$ 2-transmitters are necessary and sufficient. For simple $n$-gons the authors provided a lower bound of $\lfloor n/6 \rfloor$ 2-transmitters. We improve this bound in Section~\ref{sec:point}.


For the classical AGP variants involving guards with different capabilities have been considered; 
for example, {\it edge guards} monitor each point of the polygon that is visible to some point of the edge.
The computational complexity of the minimum edge guard problem was settled by Lee and Lin~\cite{ll-ccagp-86} who proved it to be NP-hard.
 Bjorling-Sachs~\cite{b-egrp-98} showed a tight bound of $\lfloor \frac{3n+4}{16}\rfloor$ edge guards for 
orthogonal polygons.
For general polygons $\lfloor\frac{3n}{10}  \rfloor + 1$ edge guards are always sufficient and $\lfloor \frac{n}{4} \rfloor$ are sometimes necessary~\cite{s-rrag-92}, and no tighter bounds are known.

Other problems related to $k$-transmitter coverage have also been considered. Already in 1988, Dean et al.~\cite{jls-rphs-88} considered a problem in which single edges become transparent. While for ordinary visibility the AGP equates to finding a cover of star-shaped polygons, Dean et al.~defined pseudo-star-shaped polygons to include parts that are visible through single edges. The authors concentrated on testing whether a polygon is pseudo-star-shaped, that is, whether there exists one of these more powerful guards that completely covers the input polygon.
Moreover, Mouawad and Shermer~\cite{ms-sp-94} considered the so-called {\it Superman problem}: given a polygon $P$ and its subpolygon $K$, for a point $x$ in the exterior of $P$, how many edges of $P$ must be made intransparent or opaque so that $x$ cannot see a point of $K$?

{\bf Our Results.} 
Our focus is on finding covers of lower power transmitters, that is, mainly 2-transmitters. This is in line with the work of Ballinger et al.~\cite{bbal-cktpo-10} and is motivated both by practical applications and by virtue of being the natural extension of classical Art Gallery results, that is, results for $k=0$. 

We provide NP-hardness results for several problem variants in Section~\ref{sec:np}. In Section~\ref{sec:point} we provide observations on point 2-transmitter covers and a lower bound for the number of point 2-transmitters in general polygons. We give sufficiency and necessity results for edge 2-transmitters in Section~\ref{sec:edge}; these results are summarized in Table~\ref{tab:edge}.

\begin{table}

\begin{centering}
\small
\begin{tabular}{|c|c|c|}
\hline
  Polygon Class & Always Sufficient &  Sometimes Necessary \\ \hline 
 General & $\lfloor 3n/10 \rfloor + 1 $ ~\cite{s-rrag-92} & $\lfloor n/6 \rfloor$ (Th.~\ref{thm:general-edge-lower}) \\ \hline
  Monotone  & $\lceil (n-3)/8\rceil$ (Th.~\ref{thm:m-edge-upper}) & $ \lceil (n-2)/9 \rceil $ (Th.~\ref{thm:m-edge-lower})\\ \hline
  Monotone Orthogonal  & $\lceil (n-2)/10\rceil$ (Th.~\ref{thm:mo-edge-upper})& $ \lceil (n-2)/10 \rceil $ (Th.~\ref{thm:mo-edge-lower}) \\ \hline
 Orthogonal  & $\lfloor (3n+4)/16\rfloor$~\cite{b-egrp-98} & $ \lceil (n-2)/10 \rceil $ (Th.~\ref{thm:mo-edge-lower}) \\\hline
\end{tabular}

\end{centering}

\caption{ \small A summary of results on the number of edge 2-transmitters sufficient or necessary to cover a polygon with $n$ vertices.
}
\label{tab:edge}
\vspace*{-.5cm}
\end{table}

\vspace*{-.3cm}

\section{Notations and Preliminaries.}

In a polygon $P$, a point $q\in P$ is {\it 2-visible} from $p\in \mathbb{R}^2$ if the straight-line connection $\overline{pq}$ intersects $P$ in at most two connected components. 

For a point $p\in P$, we define the {\it 2-visibility region} of $p$, 2VR($p$), as the set of points in $P$ that are 2-visible from $p$. For a set $S\subseteq P$, 2VR$(S):=\cup_{p\in S} \mbox{2VR}(p)$. A set $C\subseteq P$ is a {\em 2-transmitter cover} if 2VR$(C)=P$.

Points used for a 2-transmitter cover are called {\it (point) 2-transmitters}. By comparison, an {\it edge 2-transmitter} $e$ can monitor all points of $P$ that are 2-visible from some 
$q\in e$: 2VR$(e) =\cup_{q\in e} \mbox{2VR}(q)$. In this article, we focus on 2-transmitters in a polygon $P$ (edges or points of $P$). Most results for point 2-transmitters generalize to arbitrary locations.

Analogously, we can define $k$-visibility: a point $q\in P$ is {\it k-visible} from $p\in \mathbb{R}^2$ if the straight-line connection $\overline{pq}$ intersects $P$ in at most $k$ connected components.
For a point $p$, the {\it k-visibility region} of $p$, $k$VR($p$), is the set of points in $P$ that are $k$-visible from $p$. For a set $S\subseteq P$, $k$VR$(S):=\cup_{p\in S} k\mbox{VR}(p)$. A set $C\subseteq P$ is a {\em k-transmitter cover} if 2VR$(C)=P$.
Points used for a $k$-transmitter cover are called {\it (point) k-transmitters}.



\section{NP-hardness results}\label{sec:np}

In this section, we provide NP-hardness results for finding the minimum number of point 2-transmitters in a 2-transmitter cover in general simple and in orthogonal simple polygons, and for point $k$-transmitters in simple polygons. In addition, we show that finding the minimum number of edge 2-transmitters in an edge 2-transmitter cover is NP-hard in general simple polygons. \\

\noindent{\bf Minimum Point 2-Transmitter [$k$-transmitter] Cover (MP2TC)} \\{\bf[MP$k$TC] Problem:}\\
{\bf Given:} A polygon $P$.\\
{\bf Task:} Find the minimum cardinality 2-transmitter [$k$-transmitter] cover of $P$.
\begin{theorem}\label{th:MP2TC}
MP2TC is NP-hard for simple polygons.
\end{theorem}
\begin{proof}
We reduce from the Minimum Line Cover Problem, 
 shown to be NP-hard by Brod\'{e}n et al.~\cite{bhn-gl2lp-01}:

\noindent{\bf Minimum Line Cover Problem (MLCP)}\\
{\bf Given:} A set $L$ of non-parallel lines in the plane.\\
{\bf Task:} Find the minimum set $S$ of points such that there is at least one point in $S$ on each line in $L$.

\begin{figure}[t!]
\centering
\hspace*{.05\textwidth}
    \comic{.4\textwidth}{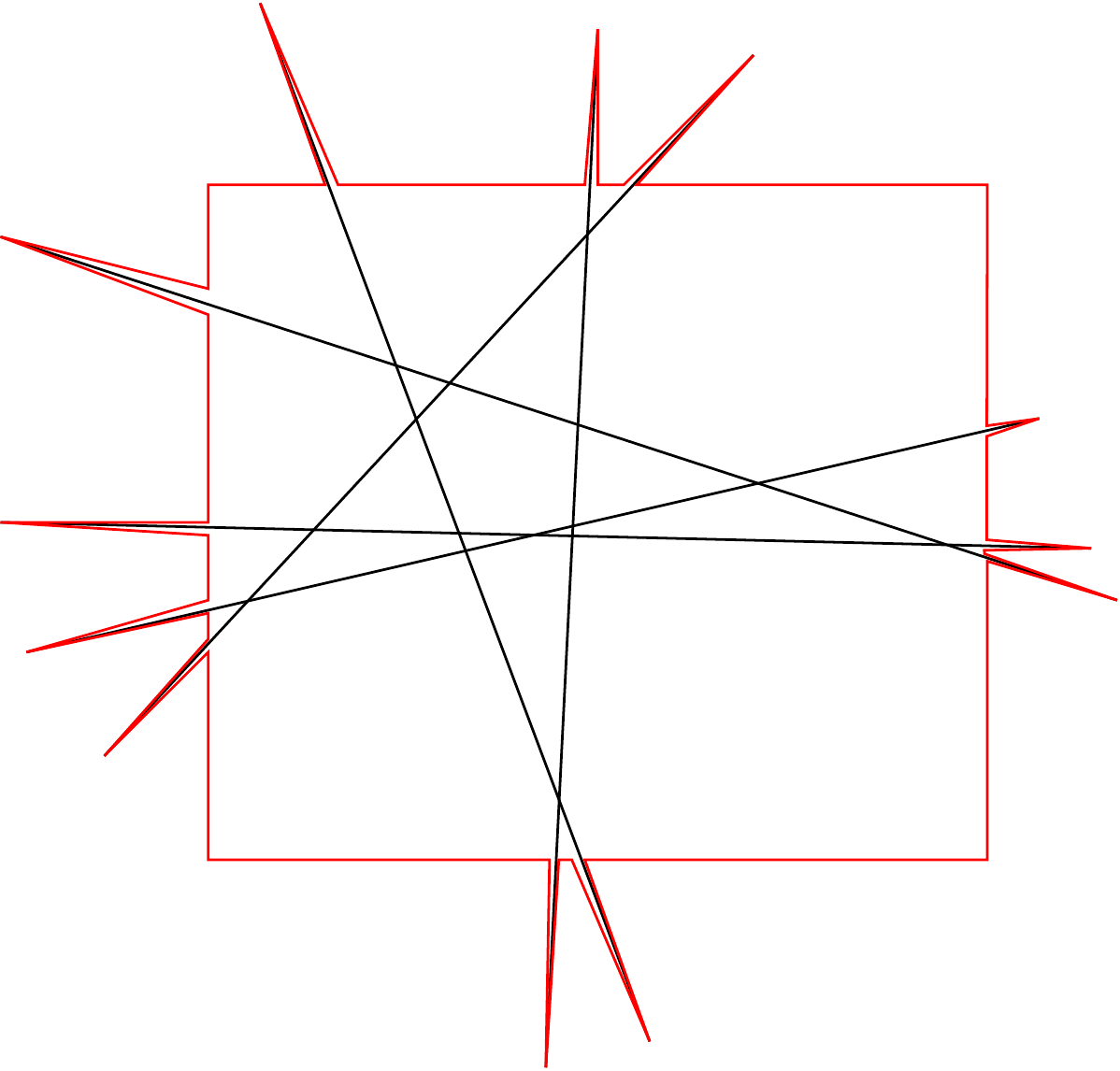}{(a)}\hfill
       \comic{.4\textwidth}{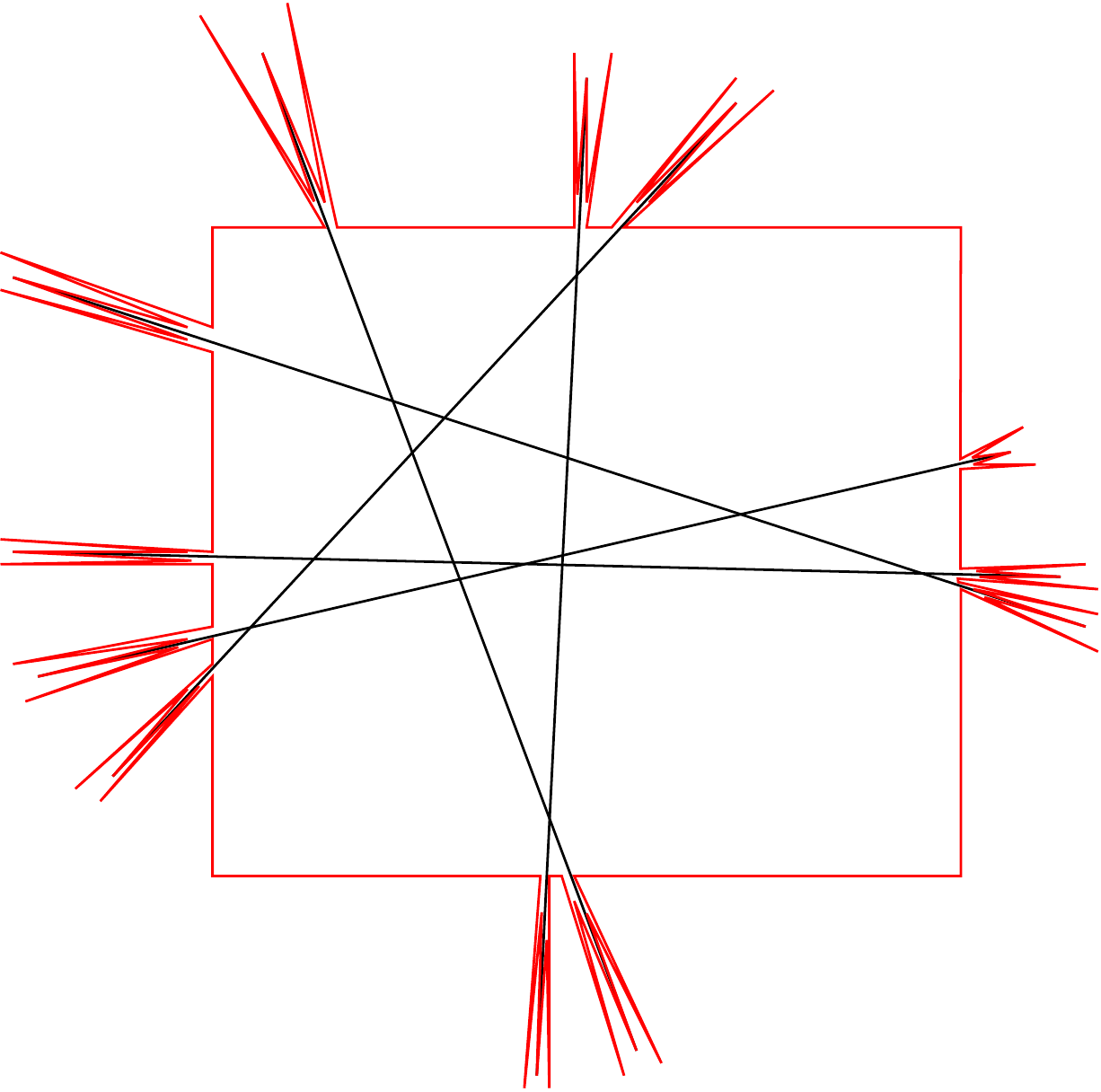}{(b)}\hspace*{.05\textwidth}
         \caption{\label{fig:np1}\small An example of our NP-hardness reduction from MLCP to MP2TC in the proof of Theorem~\ref{th:MP2TC}. (a) $L$ in black and the resulting spike box in red. (b) The additional construction for 2-transmitters. 
         }
\end{figure}

For a given set of lines $L$ we construct a ``spike box'' $P$: an (axis-aligned) square $Q$ that contains all intersection points plus two narrow spikes per line at the intersections with $Q$; see Figure~\ref{fig:np1}(a).
If we consider computing the minimum 0-transmitter cover of $P$, at least one 0-transmitter must lie on each line. Hence, this problem is equivalent to a minimum line cover.


However, for the case of 2-transmitters we must slightly modify the 
spike-box construction. At each spike we add a small ``crown'': two additional spikes, resulting in a polygon $P(L)$. The spikes ensure that points in the central spike are visible only to points in $Q$ along the original line from $L$
; see Figure~\ref{fig:np1}(b).
This yields that a minimum point cover of $L$ is equivalent to a minimum 2-transmitter cover of $P(L)$.

\end{proof}
Observe that we can easily extend the above result to point-$k$-transmitters: we enlarge the ``crowns'' and add $k/2$ spikes to each side of the central spike that relates to the input line from $L$. This yields:

\begin{corollary}
MP$k$TC is NP-hard for simple polygons.
\end{corollary}

\begin{figure}[t]
\centering
\comicII{.8\textwidth}{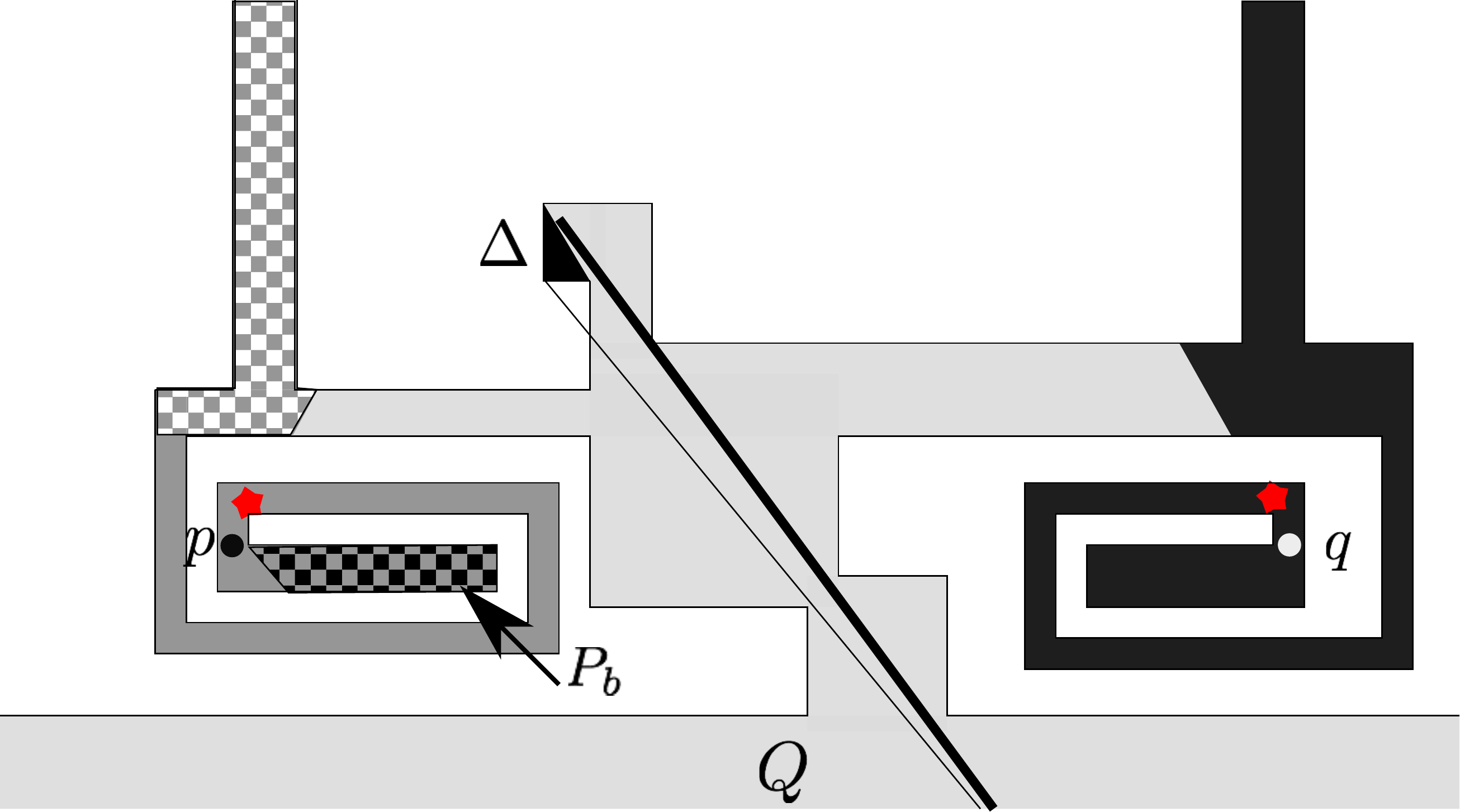}
         \caption{\label{fig:np2}\small  An orthogonal spike construction, used in the proof of Theorem~\ref{th:MP2TC-ortho}: $P(L)$ is shaded in light gray, a line of $L$ is shown in bold black.
         }
\end{figure}

\begin{theorem} \label{th:MP2TC-ortho}
MP2TC is NP-hard for orthogonal, simple polygons.
\end{theorem}

\begin{proof}

Again we prove the statement by reduction from the MLCP. However, we do not reduce from the original problem, but rather a variant that was proven to be NP-hard as well. As Biedl et al.~\cite{biikm-gp-11} showed in 2011, the MLCP remains NP-hard even if the given lines (parallelity allowed) have only one out of 4 slopes (horizontal, vertical, diagonal and off-diagonal in the octagonal grid), a problem we call MLCP4. 
The slope constraint ensures that our constructed polygon can be polynomially encoded.

Given the lines $L$ of a MLCP4 instance, we rotate all lines by 22.5$^\circ$. 
Again we construct an (axis-aligned) square $Q$ that contains all 
intersection points. Due to the rotation, no line is orthogonal to an edge of $Q$.
We construct a ``spike box'', but alter it slightly:
for each spike 
we insert the construction shown in Figure~\ref{fig:np2}
(adapted to one of the four slopes), resulting in an orthogonal polygon $P_o(L)$. 

Let $k$ be the number of lines in $L$, and $\ell$ the size of the MLCP4 solution.
%
We claim a minimum line cover of $L$ of cardinality  $\ell$ is equivalent to a minimum 2-transmitter cover of $P_o(L)$ of cardinality $\ell + 4k$.

In Figure~\ref{fig:np2},
$p$ and $q$ 
are only 2-visible from points in the regions shaded dark gray and black, respectively, and not from $Q$. 
In addition, the black triangle, labelled $\Delta$, is  2-visible 
from no other spike gadget, and in 
$Q$ only from points along the thickened line representing the line from $L$. 
If we want to place a 2-transmitter $g$ to simultaneously cover $p$ and $\Delta$, $g$ must be located in the polygon area with white squares. But then, no point in the area with black squares, $P_b$, is visible to $g$.  No point in $Q$ covers $P_b$, requiring an additional 2-transmitter within the spike gadget.
An analogous argument is used for $q$.
Thus, the minimum number of 2-transmitters that cover the spike gadget---except for $\Delta$, and parts of the lower corridors of the two spirals which are 2-visible from any point within 
$Q$---is two, located at the two (red) stars. This results in $4k$ 2-transmitters, two per spike. 

The remaining 
$\Delta$'s in all spike gadgets can be covered with $\ell$ 2-transmitters iff the lines of $L$ can be covered by $\ell$ points, which establishes the claim.

\end{proof}

\noindent{\bf Minimum Edge 2-Transmitter Cover (ME2TC):}\\
{\bf Given:} A polygon $P$.\\
{\bf Task:} Find the minimum cardinality edge 2-transmitter cover of $P$.
\begin{theorem}
ME2TC is NP-hard for simple polygons.
\end{theorem}
\begin{proof}
We adapt the proof for the minimum edge guard problem of Lee and Lin\cite{ll-ccagp-86}, i.e., their reduction from 3SAT. 

Let $F$ be an instance of the 3SAT problem. That is, $F$ is a Boolean formula in 3-CNF: it consists of a set $\mathcal{C}=\{C_1, C_2, \ldots, C_m\}$ of $m$ clauses over $n$ variables $\mathcal{V}=\{x_1, x_2, \ldots, x_n\}$, where each clause $C_i$ consists of three literals. The 3SAT problem is to decide whether there exists a truth assignment to the variables of $F$ such that each clause is satisfied.
From $F$ we construct in polynomial time a simple polygon $P$ (that is, we represent the variables and clauses by pieces of the polygon), such that $P$ admits an edge 2-transmitter cover of size $3m+n+1$ if and only if $F$ has a truth assigment satisfying all its clauses.

As with Lee and Lin's adaptation of the point guard  AGP
, we need to modify the literal pattern, the variable pattern, the vertex $W$ and in our case also the clause pattern; see Figure~\ref{fig:npedges}.
For the variable pattern we make the additional assumption that each literal appears in at least 3 clauses (or we could add spikes to the rectangle wells of the variable pattern).

In order to facilitate understanding of our proof, we give a brief recap of the point guard proof of Lee and Lin.
 They present a simple polygon that has a large central convex part, with clause gadgets on top and variable gadgets on bottom.  The clause gadgets consist of a triangular structure with three smaller literal gadgets, one per literal (variable or negated variable) that appears in the clause. 
For each literal gadget, only two points (vertices) can be used to simultaneously guard the spike of the literal gadget and a spike in the variable gadget. These two points correspond to the two possible truth settings of the literal. 
For each clause, this triangular structure is connected to the main body of the polygon, such that at least one of the point guards coinciding with a truth setting satisfying the clause must be used, as otherwise a part of the triangle cannot be guarded. 
The variable gadgets are used to enforce consistency amidst the truth settings of a variable that appears either as the variable or negated in different clauses. 
Per variable, two slightly tilted wells are added to the bottom of the polygon's main body. The wells of all variables can be monitored by the top left corner of the polygon's main body, the vertex $W$. 
For each variable pair of wells there is one triangle that can only be monitored by top points of the two wells of this variable. One well of the pair corresponds to a truth setting of true, and the other to a truth setting of false. Spikes are added to these wells that are an extension of lines of visibility into the wells from the two distinguished points of the literal pattern in the clause gadgets via the right top corner of a well. Thus, only if the truth setting of a literal is chosen consistently in all clause gadgets, all spikes of one well can be monitored from these points. In this case, only one additional guard is needed to cover the spikes of the other well, plus the triangle of the variable well pair. Thus there is a consistent satisfying assignment of truth values to variables if and only if the polygon has a guard set of size $3n + m + 1$.

For the NP-hardness proof for edge guard covers, Lee and Lin adapted these structures to ensure the properties given above. We do the same for edge $2$-transmitters.

\begin{figure*}[t!]
\centering
\hspace{1cm}
     \comic{.24\textwidth}{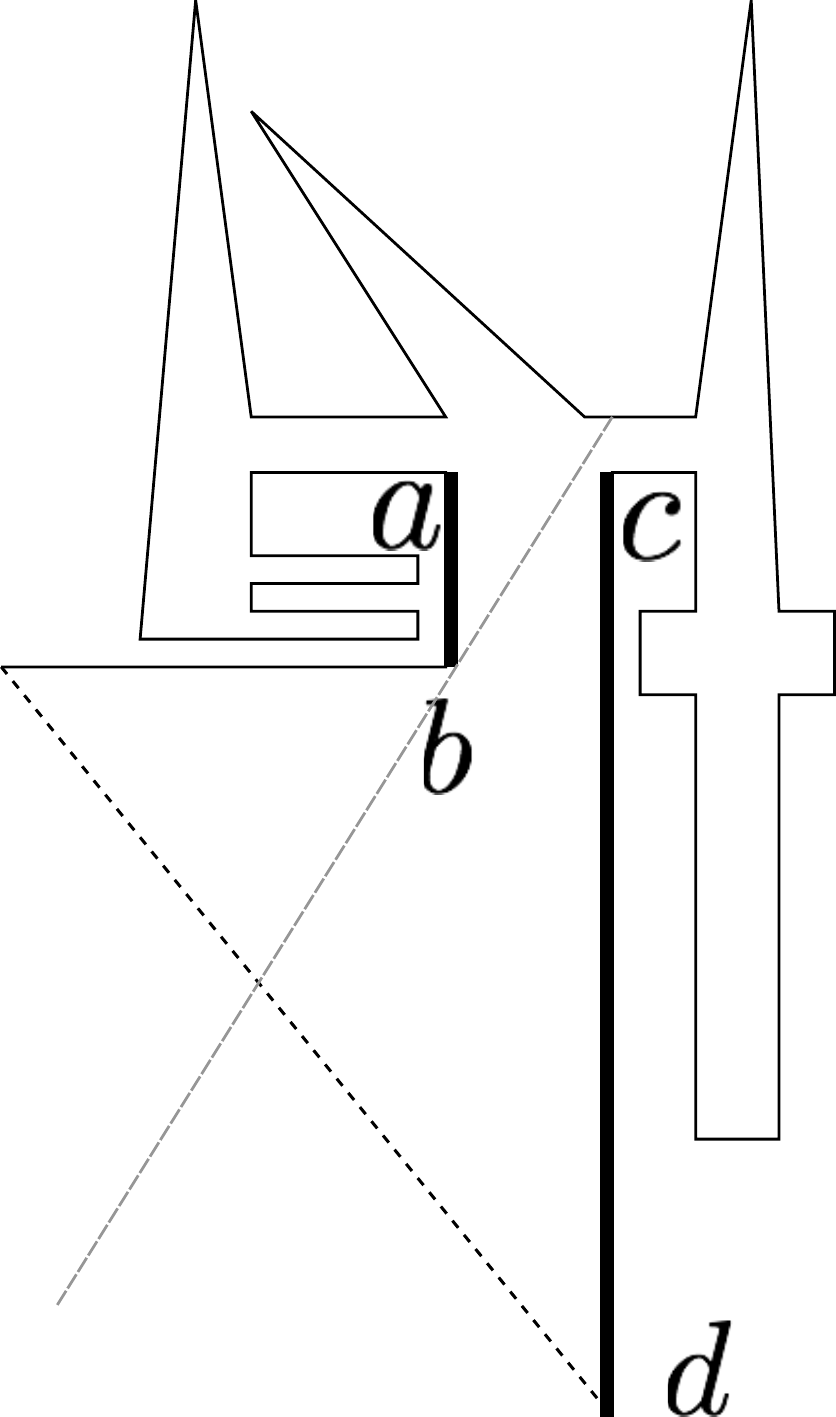}{(a)}\hspace{2cm}
       \comic{.32\textwidth}{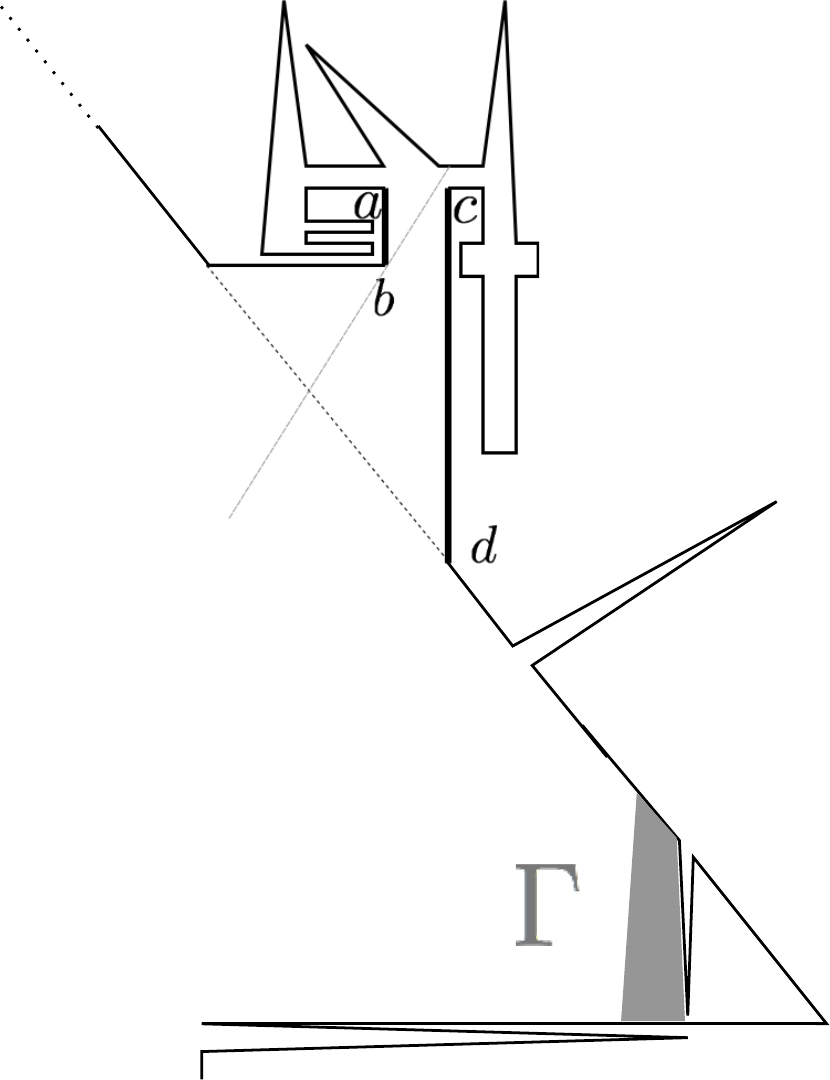}{(b)} \hspace{10cm} 
       \\ \vspace{0.5cm}
      \hfill 
	\comic{.4\textwidth}{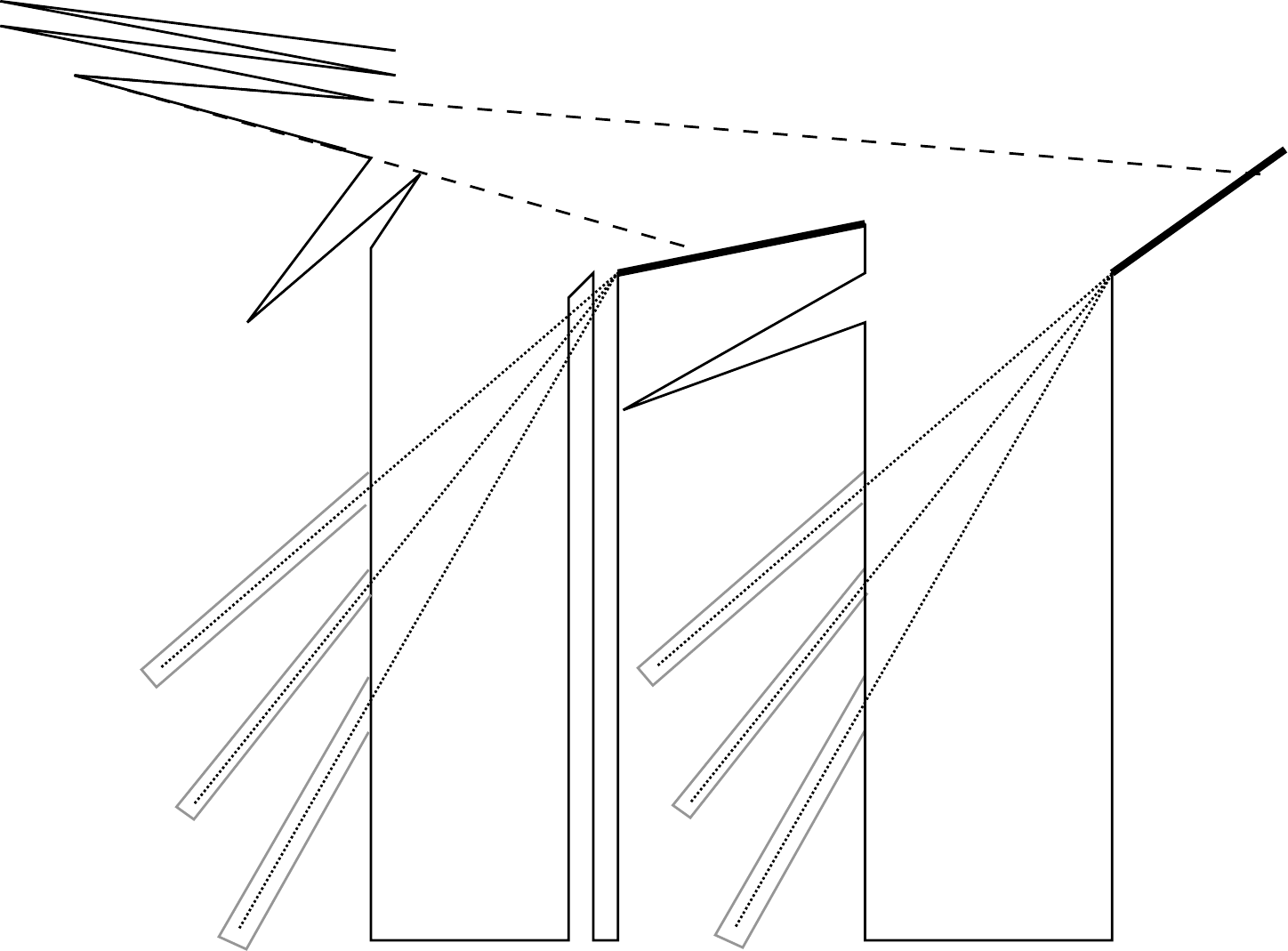}{(c)}\hfill
       \comic{.4\textwidth}{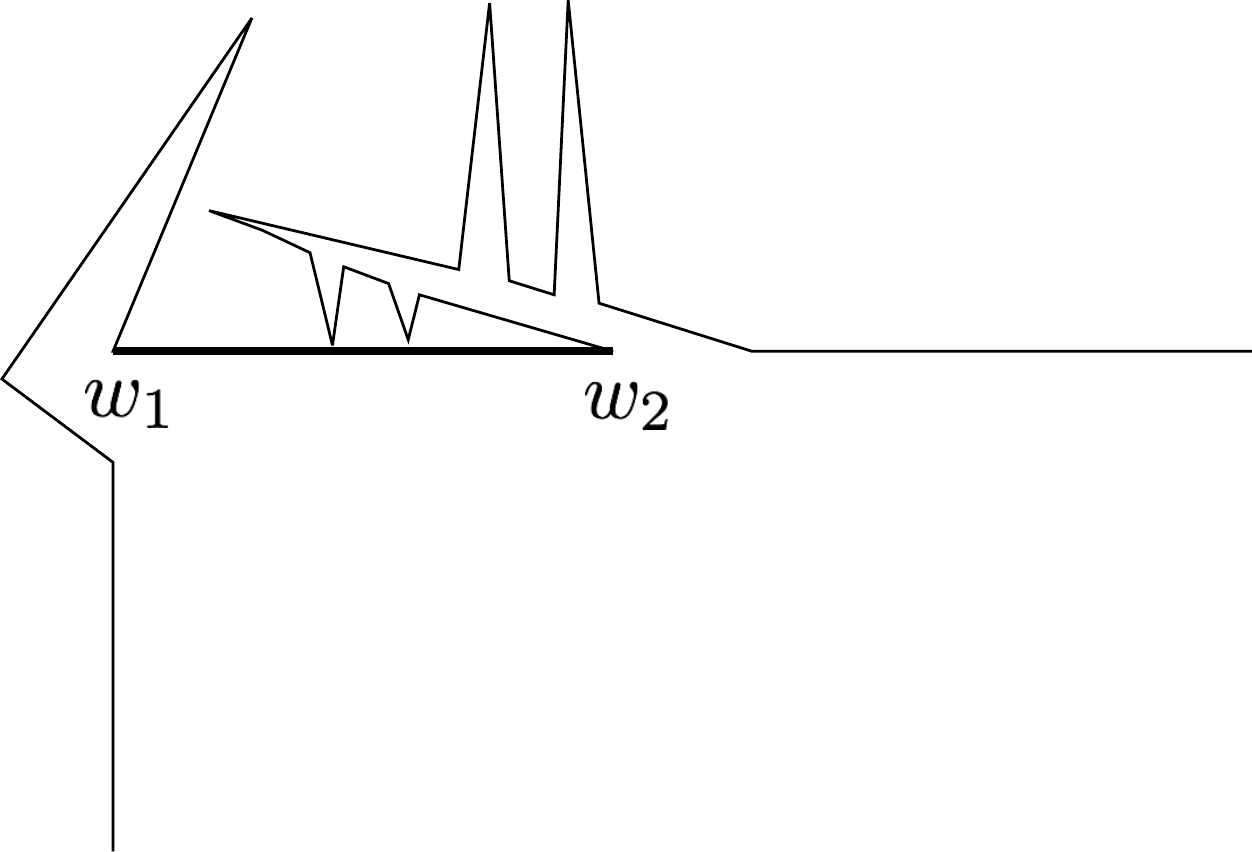}{(d)} \hfill
         \caption{\label{fig:npedges}\small (a) Literal, (b) Clause, (c) Variable, and (d) Vertex $W$  gadgets for edge 2-transmitters. Lines of sight from the variable pattern spikes are shown in gray.
       }
\end{figure*}

Figure~\ref{fig:npedges}(a) shows our literal pattern: only 
2-transmitters $\overline{ab}$ and $\overline{cd}$ 
can cover the entire literal pattern; they correspond to truth settings 
that do not and do satisfy the clause, respectively.

In the clause gadget, only the $2$-transmitters $\overline{cd}$ from the literal patterns, corresponding to a truth setting that fullfils the clause, can monitor the complete gadget. See Figure~\ref{fig:npedges}(b) for the lower right corner of the triangular clause structure: the spikes around 
region $\Gamma$ ensure (i) 
edges $\overline{ab}$ corresponding to a truth setting that does not satisfy the clause cannot 
cover
$\Gamma$, and 
(ii) no edge outside the clause gadget can cover all of
$\Gamma$
---we add the same construction on the other side of the clause 
triangle. 

Figure~\ref{fig:npedges}(c) shows the adaptation of the variable gadget. Depicted is the construction for one variable. 
Only the two bold 
2-transmitters 
cover 
the lowest of the 3 consecutive triangles to the left and also 
one of the wells and its spikes in the variable pattern.  We add two additional wells, visible from $w_1$, between variable patterns, preventing 
edges inside a variable pattern from monitoring another.

We replace vertex $W$ from the AGP construction by
edge $\overline{w_1 w_2}$, see Figure~\ref{fig:npedges}(d). 
Point $w_1$ serves the same purpose as $W$ in the original proof; only 
$\overline{w_1 w_2}$ covers both the entire attached gadget (necessary to prevent choosing edges adjacent to $W$/$\overline{w_1 w_2}$, which may cover more than desired)
 and all wells in and between variable gadgets.

Altogether, we obtain a one-to-one correspondence between edge guards from the proof of Lee and Lin and edge 2-transmitters in our constructed polygon, proving that the constructed polygon has an edge 2-transmitter cover of size $3m + n + 1$ if and only if instance $F$ of the 3SAT problem is satisfiable.

\end{proof}
\vspace{1mm}

\vspace*{-.6cm}

\section{Point 2-transmitters}\label{sec:point}

We begin with an observation that indicates 2-transmitter coverage problems have additional complexity when compared to 0-transmitter problems.  


\begin{figure}[t]
\center
\includegraphics[scale = 0.5]{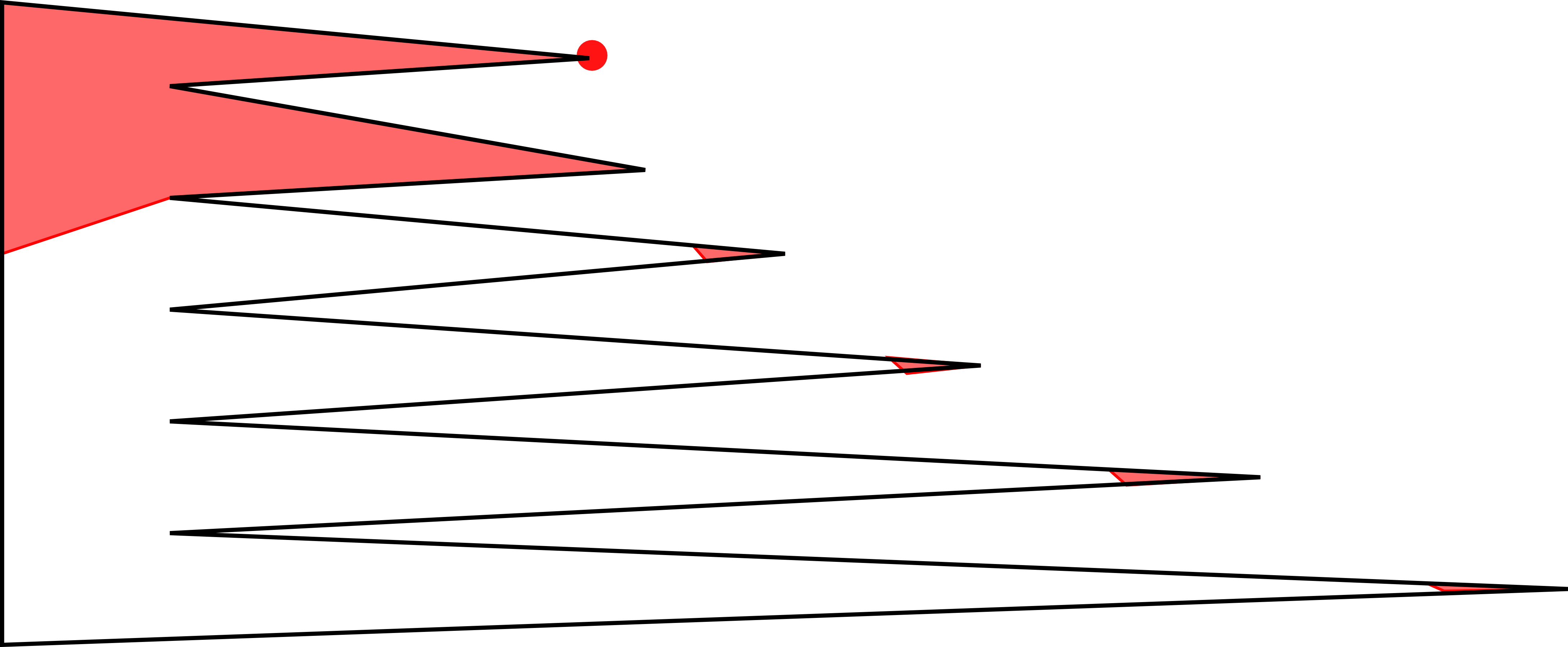}
         \caption{\small \label{fig:pt}
 The 2-visibility region of the red point (shaded light red) has $O(n)$ connected components.
         }
\end{figure}

\begin{observation}
In a simple polygon $P$, the  2-visibility region of a single guard can have $O(n)$ connected components. See Figure~\ref{fig:pt}(b).
\end{observation}

We now present lemmas on point 2-transmitter covers that enable the edge 2-transmitter results. 

\begin{lemma}\label{lem:5gon}
Every 5-gon 
can be covered by a point 2-transmitter placed anywhere 
(boundary or interior).
\end{lemma}
\begin{proof}
Any 5-gon is 2-convex (i.e., the intersection of any line with a 5-gon $P$ has at most two connected components). 
\end{proof}

\begin{lemma}\label{lem:6gon}
Let $P$ be a 6-gon and let $e=\{v,w\}$ an edge of $P$. A point 2-transmitter at  $v$ or at $w$ covers $P$.


\end{lemma}

\begin{proof}
By the Two Ear Theorem \cite{m-phe-75}, there 
exists a diagonal from either $v$ or $w$ that splits off a triangle $T$ from $P$. Removing $T$ leaves a 
5-gon $P'$ which has $v$ or $w$ as one of its vertices; denote this vertex $\tilde{v}$. The addition of 
$T$ does not increase the number of connected components 
 of the polygon on any visibility ray starting at $\tilde{v}$. That is, for any ray $r$ starting at $\tilde{v}$, the number of connected components 
 of $r\cap P'$ is the same as the number of connected components 
 of $r \cap P$. 
Thus, all of $P$ is visible from a 
$2$-transmitter on $\tilde{v}$.
\end{proof}

\begin{lemma}\label{lem:monotone6gon}
Every monotone 6-gon can be covered by a single (point) 2-transmitter placed at one of its two leftmost (or rightmost) vertices.
\end{lemma}

\begin{proof}
The statement follows directly from Lemma~\ref{lem:6gon}, as in a monotone polygon, the leftmost vertex $v_1$ and the second-to-leftmost vertex $v_2$ are adjacent.
\end{proof}

The following Splitting Lemma was shown by Aichholzer et al.~\cite{affhhuv-mimp-09}. As it plays a central role in one of our proofs for edge 2-transmitters, we present the lemma and a sketch of its proof:

\begin{lemma}
[Splitting Lemma, \cite{affhhuv-mimp-09}] \label{lem:split}
Let $P$ be a monotone polygon with vertices $p_1$, $p_2$,..., $p_n$, ordered from left to right. For every positive integer $m<n$, there exists a vertical line segment $l$ and two monotone polygons $L$ and $R$ such that 
\begin{dense_itemize}
\item $L$ has $m$ vertices and $R$ has $n-m+2$ vertices.
\item Either $l$ is a chord of $L$ and an edge of $R$, or $l$ is an edge of $L$ and a chord of $R$.
\item $p_m$ or $p_{m+1}$ is an endpoint of $l$.
\item Denote as $L'$ the subset of $L$ left of $l$ and denote as $R'$ the subset of $R$ right of $l$; then $P = L' \cup R'$.
\end{dense_itemize}
\end{lemma}
{\it Proof Sketch.} Consider a vertical line intersecting $P$ between $p_{m-1}$ and $p_m$.  The 
edges $e$ and $f$ this vertical line crosses, when extended, meet to either the right or left of this line (assuming they are not parallel). If they meet to the left,  then $l$ is the vertical segment in ${\rm int}(P)$ containing $p_{m-1}$, $L$ is the portion of $P$ left of $l$, and $R$ is the portion of $P$ right of $l$ together with the triangle formed by $l$ and the extensions of $e$ and $f$; see Figure~\ref{fig:split} for an example. 
If $e$ and $f$ meet to the right, $l$ runs through $p_m$ and the construction rules of $R$ and $L$ are swapped.
Careful analysis shows 
$R$ and $L$ satisfy the 
stated properties.

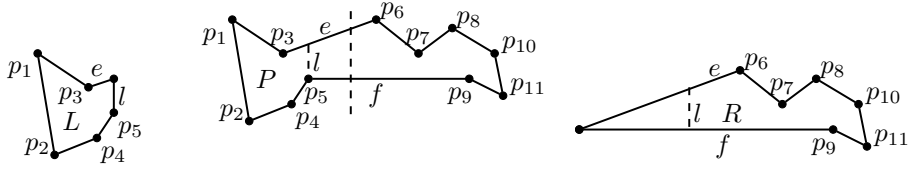
\begin{figure}

\begin{centering}

{\small
\begin{tikzpicture}[scale=0.45]
	\begin{pgfonlayer}{nodelayer}
		\node [style=vertex] (0) at (-3.25, 0.2499996) {};
		\node [style=vertex] (1) at (-2.75, -2.749999) {};
		\node [style=vertex] (2) at (-1.75, -0.7499997) {};
		\node [style=vertex] (3) at (-1.5, -2.25) {};
		\node [style=vertex] (4) at (4.749999, -2) {};
		\node [style=vertex] (5) at (3.75, -1.499999) {};
		\node [style=vertex] (6) at (-1, -1.5) {};
		\node [style=vertex] (7) at (0.9999987, 0.2500002) {};
		\node [style=vertex] (8) at (2.25, -0.7499987) {};
		\node [style=vertex] (9) at (4.499999, -0.7499987) {};
		\node [style=vertex] (10) at (3.249999, 0) {};
		\node [style=none] (11) at (0.25, 0.5) {};
		\node [style=none] (12) at (0.25, -2.75) {};
		\node [style=none] (13) at (-0.5000001, -0) {$e$};
		\node [style=none] (14) at (1, -2) {$f$};
		\node [style=none] (15) at (-0.7499997, -1) {$l$};
		\node [style=none] (16) at (-1, -1.5) {};
		\node [style=vertex] (17) at (-7.25, -3.25) {};
		\node [style=vertex] (18) at (-8.5, -3.75) {};
		\node [style=none] (19) at (-6.75, -2.5) {};
		\node [style=vertex] (20) at (-9, -0.75) {};
		\node [style=vertex] (21) at (-6.75, -2.5) {};
		\node [style=vertex] (22) at (-7.5, -1.75) {};
		\node [style=none] (23) at (-7.25, -1.25) {$e$};
		\node [style=none] (24) at (-6.5, -2) {$l$};
		\node [style=vertex] (25) at (-6.75, -1.5) {};
		\node [style=vertex] (26) at (14.5, -3) {};
		\node [style=vertex] (27) at (15.5, -3.5) {};
		\node [style=vertex] (28) at (14, -1.5) {};
		\node [style=none] (29) at (10.5, -2.5) {$l$};
		\node [style=vertex] (30) at (13, -2.25) {};
		\node [style=vertex] (31) at (15.25, -2.25) {};
		\node [style=vertex] (32) at (11.75, -1.25) {};
		\node [style=vertex] (33) at (7, -3) {};
		\node [style=none] (34) at (11.25, -3.5) {$f$};
		\node [style=none] (35) at (11, -1.25) {$e$};
		\node [style=vertex] (36) at (7, -3) {};
		\node [style=none] (37) at (10.25, -3) {};
		\node [style=none] (38) at (-8, -2.75) {$L$};
		\node [style=none] (39) at (11.5, -2.5) {$R$};
		\node [style=none] (40) at (-2.25, -1.5) {$P$};
		\node [style=none] (41) at (-1, -0.5000001) {};
		\node [style=none] (42) at (10.25, -1.75) {};
		\node [style=none] (43) at (1.5, 0.5) {$p_6$};
		\node [style=none] (44) at (-0.7499997, -2) {$p_5$};
		\node [style=none] (45) at (-6.25, -3) {$p_5$};
		\node [style=none] (46) at (12.25, -1) {$p_6$};
		\node [style=none] (47) at (-3.75, -0.2499996) {$p_1$};
		\node [style=none] (48) at (-3.25, -2.5) {$p_2$};
		\node [style=none] (49) at (-1.75, -0.2499996) {$p_3$};
		\node [style=none] (50) at (-1, -2.749999) {$p_4$};
		\node [style=none] (51) at (-6.75, -3.75) {$p_4$};
		\node [style=none] (52) at (-8, -2) {$p_3$};
		\node [style=none] (53) at (-9.5, -1.25) {$p_1$};
		\node [style=none] (54) at (-9, -3.5) {$p_2$};
		\node [style=none] (55) at (2.25, -0.2499999) {$p_7$};
		\node [style=none] (56) at (3.75, 0.25) {$p_8$};
		\node [style=none] (57) at (3.5, -2) {$p_9$};
		\node [style=none] (58) at (5.249999, -0.4999995) {$p_{10}$};
		\node [style=none] (59) at (5.5, -1.75) {$p_{11}$};
		\node [style=none] (60) at (16, -2) {$p_{10}$};
		\node [style=none] (61) at (16.25, -3.25) {$p_{11}$};
		\node [style=none] (62) at (14.25, -3.5) {$p_9$};
		\node [style=none] (63) at (13, -1.75) {$p_7$};
		\node [style=none] (64) at (14.5, -1.25) {$p_8$};
	\end{pgfonlayer}
	\begin{pgfonlayer}{edgelayer}
		\draw [style=simple] (0) to (1);
		\draw [style=simple] (1) to (3);
		\draw [style=simple] (3) to (6);
		\draw [style=simple] (6) to (5);
		\draw [style=simple] (5) to (4);
		\draw [style=simple] (9) to (10);
		\draw [style=simple] (10) to (8);
		\draw [style=simple] (8) to (7);
		\draw [style=simple] (7) to (2);
		\draw [style=simple] (2) to (0);
		\draw [style=dash] (11.center) to (12.center);
		\draw [style=simple] (9) to (4);
		\draw [style=simple] (20) to (18);
		\draw [style=simple] (18) to (17);
		\draw [style=simple] (17) to (21);
		\draw [style=simple] (25) to (22);
		\draw [style=simple] (22) to (20);
		\draw [style=simple] (19.center) to (25);
		\draw [style=simple] (26) to (27);
		\draw [style=simple] (31) to (28);
		\draw [style=simple] (28) to (30);
		\draw [style=simple] (30) to (32);
		\draw [style=simple] (31) to (27);
		\draw [style=simple] (33) to (26);
		\draw [style=simple] (32) to (36);
		\draw [style=dash] (41.center) to (6);
		\draw [style=dash] (42.center) to (37.center);
	\end{pgfonlayer}
\end{tikzpicture} }

\end{centering}
\caption{\small Application of the Splitting Lemma of \cite{affhhuv-mimp-09} to an 11-gon $P$ for $m = 6$, resulting in a 6-gon $L$ and a 7-gon $R$.}
\label{fig:split}
\end{figure}

\subsection{Necessity for General Polygons}

We recall that for general polygons and standard (0-transmitting)  guards, $\lfloor n/3 \rfloor$ guards are sometimes necessary and  always sufficient to guard a polygon with $n$ vertices \cite{c-ctpg-75}. It follows that for  2-transmitters, $\lfloor n/3 \rfloor$ are also always sufficient to cover a polygon with $n$ vertices by simply considering 2-transmitters as weaker 0-transmitters instead. 
We demonstrate that $\lfloor n/5 \rfloor$ 2-transmitters are sometimes necessary to cover a polygon with $n$ vertices, improving the $\lfloor n/6 \rfloor$ bound of \cite{bbal-cktpo-10}. 

\begin{figure}[b!]

\begin{centering}

\begin{tikzpicture}[scale=0.7]

\begin{pgfonlayer}{nodelayer}

\node [style=vertex] (0) at (-1.25, 0.75) {};

\node [style=none] (1) at (5, 0.5) {};

\node [style=vertex] (2) at (2.75, 1.25) {};

\node [style=vertex] (3) at (-2.25, 1) {};

\node [style=bigvertex] (4) at (1, 1.5) {};

\node [style=none] (5) at (-6.5, -1) {};

\node [style=vertex] (6) at (2, 1.25) {};

\node [style=none] (7) at (0, 0.5) {};

\node [style=none] (8) at (1.25, 0.25) {};

\node [style=none] (9) at (1.5, 1.75) {$a$};

\node [style=none] (10) at (-6, -2) {};

\node [style=none] (11) at (-4.75, -1) {};

\end{pgfonlayer}

\begin{pgfonlayer}{edgelayer}

\fill[lightgray]  (-3.5,0)--(-2.25,1)--(2,1.25)--(1,1.5)--(2.75,1.25)--(-1.25,.75)-- cycle;	

\draw [style=simple] (4) to (2);
		\draw [style=simple] (2) to (0);
		\draw [style=simple] (4) to (6);
		\draw [style=graydash] (4) to (1.center);
		\draw [style=graydash] (4) to (5.center);
		\draw [style=simple] (6) to (3);
		\draw [style=dasharrow] (7.center) to (8.center);
		\draw [style=simple] (0) to (7.center);
		\draw [style=dasharrow] (11.center) to (10.center);
		\draw [style=simple] (3) to (11.center);
	\end{pgfonlayer}
\end{tikzpicture}

\end{centering}

\caption{\small Lower bound construction for point 2-transmitters in general polygons.}
\label{fig:point-lower-gadget}

\end{figure}
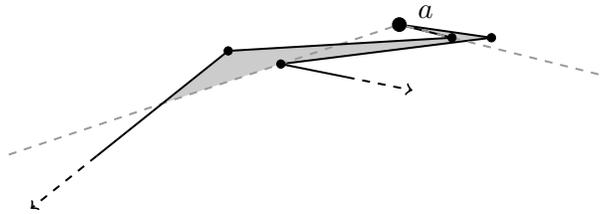

\begin{theorem}\label{thm:general-point-lower}
There exist simple polygons on $n$ vertices that require $\lfloor n/5 \rfloor$ 2-transmitters to cover their interior. 
\end{theorem}

\begin{proof}
We provide an example of a polygon $P_n$ where any valid covering requires at least $\lfloor n/5 \rfloor$ 2-transmitters. Figure~\ref{fig:point-lower-gadget} depicts a five-vertex gadget. Dashed arrows indicate edges that connect to neighboring gadgets. One can arrange  $\lfloor n/5\rfloor$ of these gadgets sequentially around a circle, subdividing up to 4 edges if necessary so that $P_n$ has exactly $n$ vertices and $n$ edges. Modified angles allow for an arbitrary large number to be placed around a circle.

Vertex $a$ is only 2-visible from the shaded region (if all other gadgets remain entirely below both gray dashed 2-visibility lines). For distinct gadgets these shaded regions are disjoint. Thus at least one point from each gadget must be included in any valid 2-transmitter cover. 
 \end{proof}

\vspace*{-.2cm}

\section{Edge 2-transmitters}\label{sec:edge}

Analogous to the direction of inquiry taken for traditional (0-transmitting) art gallery guards, in this section we consider the natural generalization of point 2-transmitters to edge 2-transmitters. 
We first consider general polygons, 
and then look at monotone and 
monotone orthogonal polygons. 
\subsection{General Polygons}

For general $n$-gons, 
the upper bound of $\lfloor \frac{3n}{10} \rfloor +1$ edge guards from~\cite{s-rrag-92} 
obviously 
holds for more powerful edge 2-transmitters.
The polygon requiring $\lfloor \frac{n}{4} \rfloor$ edge guards only necessitates $\lfloor \frac{n}{8} \rfloor$ edge 2-transmitters, and next we improve on this lower bound:

\begin{theorem} \label{thm:general-edge-lower}
There exist simple $n$-gons that require $\lfloor n/6 \rfloor$ edge 2-transmitters.
\end{theorem}

\begin{proof}
We provide 
a polygon $P_n$ where any valid covering by edge 2-transmitters requires at least $\lfloor n/6 \rfloor$ edges. Figure~\ref{fig:edge-lower-gadget}(a) depicts a six-edge gadget. The dashed arrows indicate the beginnings of edges of neighboring gadgets; one can arrange  $\lfloor n/6\rfloor$ of these gadgets sequentially around a circle, subdividing up to 5 edges if necessary so that $P_n$ has $n$ vertices and $n$ edges. 
(Modified angles allow for an arbitrarily large number to be placed around a circle.)
Vertex $a$ 
is only 2-visible from one of the six edges in the gadget 
(if all other gadgets remain entirely below the gray dashed 2-visibility lines).
In addition, $a$ is not visible from the endpoints of the gadget edges that are incident to edges from adjacent gadgets.
Thus, at least one edge from each gadget must be included in any valid edge 2-transmitter cover.
\end{proof}

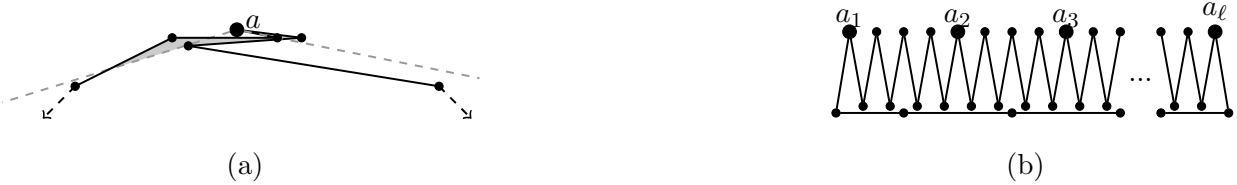
\begin{figure}[t!]

\begin{centering}
\begin{tikzpicture}[scale=0.43]
	\begin{pgfonlayer}{nodelayer}
\node [style=none] (38) at (-1, -2) {(a)};		

\node [style=vertex] (0) at (-2.75, 1.75) {};

\node [style=vertex] (1) at (5, 0.5) {};

\node [style=none] (2) at (6.25, 0.7499998) {};

\node [style=vertex] (3) at (0.75, 2) {};

\node [style=vertex] (4) at (-3.25, 2) {};

\node [style=bigvertex] (5) at (-1.25, 2.25) {};

	\node [style=none] (6) at (-6.25, 0.5) {};

	\node [style=none] (7) at (-8.5, 0) {};

\node [style=vertex] (8) at (-6.25, 0.5) {};

\node [style=vertex] (9) at (0, 2) {};

\node [style=none] (10) at (-7.25, -0.5) {};

	\node [style=none] (11) at (5, 0.5) {};

	\node [style=none] (12) at (6, -0.5) {};

		\node [style=none] (13) at (-0.75, 2.5) {$a$};

\end{pgfonlayer}
	\begin{pgfonlayer}{edgelayer}
	\fill[lightgray]  (-5.23,1.01)--(-3.25,2)--(0,2)--(-1.25,2.25)--(0.75,2)--(-2.80,1.75)-- cycle;

	\draw [style=simple] (5) to (3);
		\draw [style=simple] (3) to (0);
		\draw [style=simple] (0) to (1);
		\draw [style=simple] (5) to (9);
		\draw [style=graydash] (5) to (2.center);
		\draw [style=graydash] (5) to (7.center);
		\draw [style=simple] (9) to (4);
		\draw [style=simple] (4) to (8);
		\draw [style=simple] (8) to (6.center);
		\draw [style=simple] (1) to (11.center);
		\draw [style=dasharrow] (11.center) to (12.center);
		\draw [style=dasharrow] (6.center) to (10.center);
	\end{pgfonlayer}
\end{tikzpicture}\hfill
\begin{tikzpicture}[scale=0.36]
	\begin{pgfonlayer}{nodelayer}
	\node [style=none] (30) at (0, -2) {(b)};
		\node [style=vertex] (0) at (-1.5, 3) {};
		\node [style=vertex] (1) at (-2, 0.25) {};
		\node [style=vertex] (2) at (-3.5, 3) {};
		\node [style=bigvertex] (3) at (-2.5, 3) {};
		\node [style=vertex] (4) at (-3, 0.25) {};
		\node [style=vertex] (5) at (-4, 0.25) {};
		\node [style=vertex] (6) at (-5, 0.25) {};
		\node [style=vertex] (7) at (-6, 0.25) {};
		\node [style=vertex] (8) at (-4.5, 3) {};
		\node [style=vertex] (9) at (-5.5, 3) {};
		\node [style=vertex] (10) at (-1, 0.25) {};
		\node [style=vertex] (11) at (0, 0.25) {};
		\node [style=vertex] (12) at (1, 0.25) {};
		\node [style=vertex] (13) at (-0.5, 3) {};
		\node [style=vertex] (14) at (0.5, 3) {};
		\node [style=bigvertex] (15) at (1.5, 3) {};
		\node [style=bigvertex] (16) at (-6.5, 3) {};
		\node [style=vertex] (17) at (-0.5, -0) {};
		\node [style=vertex] (18) at (-4.5, -0) {};
		\node [style=vertex] (19) at (-7, -0) {};
		\node [style=vertex] (20) at (3.5, -0) {};
		\node [style=none] (21) at (-6.5, 3.5) {$a_1$};
		\node [style=none] (22) at (-2.5, 3.5) {$a_2$};
		\node [style=none] (23) at (1.5, 3.5) {$a_3$};
		\node [style=vertex] (24) at (2, 0.25) {};
		\node [style=vertex] (25) at (3, 0.25) {};
		\node [style=vertex] (26) at (2.5, 3) {};
		\node [style=vertex] (27) at (3.5, 3) {};
		\node [style=vertex] (28) at (5, 3) {};
		\node [style=vertex] (29) at (5.5, 0.25) {};
		\node [style=vertex] (30) at (5, -0) {};
		\node [style=vertex] (31) at (6, 3) {};
		\node [style=vertex] (32) at (6.5, 0.25) {};
		\node [style=vertex] (33) at (7.5, -0) {};
		\node [style=bigvertex] (34) at (7, 3) {};
		\node [style=none] (35) at (7, 3.75) {$a_\ell$};
		\node [style=none] (36) at (4.25, 1.25) {...};
	\end{pgfonlayer}
	\begin{pgfonlayer}{edgelayer}
		\draw [style=simple] (3) to (1);
		\draw [style=simple] (1) to (0);
		\draw [style=simple] (3) to (4);
		\draw [style=simple] (4) to (2);
		\draw [style=simple] (16) to (7);
		\draw [style=simple] (7) to (9);
		\draw [style=simple] (9) to (6);
		\draw [style=simple] (6) to (8);
		\draw [style=simple] (8) to (5);
		\draw [style=simple] (5) to (2);
		\draw [style=simple] (0) to (10);
		\draw [style=simple] (10) to (13);
		\draw [style=simple] (13) to (11);
		\draw [style=simple] (11) to (14);
		\draw [style=simple] (14) to (12);
		\draw [style=simple] (12) to (15);
		\draw [style=simple] (16) to (19);
		\draw [style=simple] (19) to (18);
		\draw [style=simple] (18) to (17);
		\draw [style=simple] (17) to (20);
		\draw [style=simple] (15) to (24);
		\draw [style=simple] (24) to (26);
		\draw [style=simple] (26) to (25);
		\draw [style=simple] (25) to (27);
		\draw [style=simple] (28) to (29);
		\draw [style=simple] (29) to (31);
		\draw [style=simple] (31) to (32);
		\draw [style=simple] (32) to (34);
		\draw [style=simple] (34) to (33);
		\draw [style=simple] (33) to (30);
	\end{pgfonlayer}
\end{tikzpicture}

\end{centering}

\caption{\small  Lower bound construction for edge 2-transmitters for (a) general and (b) monotone polygons.}
\label{fig:edge-lower-gadget}

\end{figure}

\subsection{Monotone Polygons}

\begin{theorem} \label{thm:m-edge-lower}
There exist monotone $n$-gons that require $\lceil (n-2) /9 \rceil$ edge 2-transmitters.
\end{theorem}
\begin{proof}
Consider polygon $P$ in Figure~\ref{fig:edge-lower-gadget}(b).
Two points $p_i, p_j\in {\rm int}(P)$ within $\varepsilon$ from $a_i, a_j$
are not  2-visible from the same 
edge.
So, each $a_i$ requires an edge for coverage. 
$|V(P)|=9\ell - 6$ and  $\ell$ edge 2-transmitters are necessary. 
For any value of $n$, 
the construction (possibly with up to 8 edges subdivided) necessitates $\lfloor (n+6)/9 \rfloor = \lceil (n-2)/9 \rceil$ edge 2-transmitters. 
\end{proof}

Before we can prove the upper bound for monotone $n$-gons, we present a crucial lemma.

\begin{lemma}\label{lem:monotone10gon}
Any monotone 10-gon $P$ can be covered by a single edge 2-transmitter $e$, and for every point $p \in P$, there exists $q$ on $e$ such that $p$ is 2-visible from $q$, where $q$ is left of at least two vertices of $P$ and right of at least two vertices of~$P$. 
\end{lemma}

\begin{figure}[t]
\centering
{
\begin{tikzpicture}[scale=0.6]
	\begin{pgfonlayer}{nodelayer}
		\node [style=vertex] (0) at (-0.75, -1.75) {};
		\node [style=vertex] (1) at (-1, 0.75) {};
		\node [style=vertex] (2) at (1, -0.5) {};
		\node [style=vertex] (3) at (-0.25, -1) {};
		\node [style=vertex] (4) at (2.75, 1.25) {};
		\node [style=vertex] (5) at (5.25, 0.25) {};
		\node [style=vertex] (6) at (3.75, -1) {};
		\node [style=none] (7) at (6, 1.5) {};
		\node [style=vertex] (8) at (4, -0.25) {};
		\node [style=none] (9) at (2.75, 1.5) {};
		\node [style=none] (10) at (2.75, -1.5) {};
		\node [style=none] (11) at (1, 1.5) {};
		\node [style=none] (12) at (1, -1.5) {};
		\node [style=vertex] (13) at (5.75, -0.75) {};
		\node [style=none] (14) at (1.25, -1.75) {$l_4$};
		\node [style=none] (15) at (3, -1.75) {$l_5$};
		\node [style=none] (16) at (-0.25, -0.25) {$P_L$};
		\node [style=none] (17) at (4.25, 0.25) {$P_R$};
		\node [style=none] (18) at (-1, 1.25) {$v_1$};
		\node [style=none] (19) at (-1.25, -1.75) {$v_2$};
		\node [style=none] (20) at (0.5, 0.25) {$v_4$};
		\node [style=none] (21) at (0, -1.5) {$v_3$};
		\node [style=none] (22) at (3.25, 1.5) {$v_5$};
		\node [style=none] (23) at (4, -1.5) {$v_6 $};
		\node [style=none] (24) at (4.25, -0.75) {$v_7$};
		\node [style=none] (25) at (6.25, -1) {$v_9$};
		\node [style=none] (26) at (5, 0.75) {$v_8$};
		\node [style=none] (27) at (2, -0.75) {$e$};
		\node [style=none] (28) at (6, 1.5) {};
		\node [style=none] (29) at (3.5, 0.5) {$f$};
		\node [style=none] (30) at (6.5, 1.25) {$v_{10}$};
		\node [style=vertex] (31) at (6, 1.5) {};
		\node [style=none] (32) at (1.75, 0.75) {$e'$};
		\node [style=none] (33) at (-7.5, -1.5) {};
		\node [style=none] (34) at (-3.75, -1.75) {$v_9$};
		\node [style=none] (35) at (-3.75, 0.75) {};
		\node [style=none] (36) at (-11.5, -1.75) {$v_1$};
		\node [style=vertex] (37) at (-9.75, -0) {};
		\node [style=vertex] (38) at (-10.5, 0.75) {};
		\node [style=none] (39) at (-7.75, -1.75) {$l_5$};
		\node [style=none] (40) at (-9.5, 0.5) {$v_3$};
		\node [style=none] (41) at (-9, -1.5) {};
		\node [style=none] (42) at (-6.25, 1.25) {$v_6$};
		\node [style=none] (43) at (-9, 1.5) {};
		\node [style=none] (44) at (-9.25, -1.75) {$l_4$};
		\node [style=vertex] (45) at (-7.5, 0.25) {};
		\node [style=none] (46) at (-8.5, -1) {$v_4$};
		\node [style=vertex] (47) at (-5.75, 0.25) {};
		\node [style=none] (48) at (-5.75, -1.5) {};
		\node [style=none] (49) at (-11, 1) {$v_2$};
		\node [style=vertex] (50) at (-4.25, -1.75) {};
		\node [style=none] (51) at (-8, 0.5) {$v_5$};
		\node [style=vertex] (52) at (-4.75, -0.5) {};
		\node [style=none] (53) at (-7, -0.25) {$e$};
		\node [style=vertex] (54) at (-3.75, 0.75) {};
		\node [style=none] (55) at (-5.75, 1.5) {};
		\node [style=vertex] (56) at (-11, -1.75) {};
		\node [style=none] (57) at (-3.25, 0.5) {$v_{10}$};
		\node [style=vertex] (58) at (-6.75, 1.25) {};
		\node [style=none] (59) at (-5, -1) {$v_8$};
		\node [style=none] (60) at (-5.25, 0.75) {$v_7$};
		\node [style=none] (61) at (-6, -1.75) {$l_7$};
		\node [style=none] (62) at (-10.25, -0.75) {$P_L$};
		\node [style=vertex] (63) at (-9, -0.5) {};
		\node [style=none] (64) at (-7.5, 1.5) {};
		\node [style=none] (65) at (-3.75, 0.75) {};
		\node [style=none] (66) at (-4.5, -0) {$P_R$};
		\node [style=none] (67) at (2.5, -2.75) {(b)};
		\node [style=none] (68) at (-7.75, -2.75) {(a)};
	\end{pgfonlayer}
	\begin{pgfonlayer}{edgelayer}
		\draw [style=simple] (1) to (0);
		\draw [style=simple] (0) to (3);
		\draw [style=simple] (3) to (6);
		\draw [style=simple] (6) to (8);
		\draw [style=simple] (1) to (2);
		\draw [style=wide] (2) to (4);
		\draw [style=simple] (4) to (5);
		\draw [style=simple] (5) to (7.center);
		\draw [style=dash] (11.center) to (12.center);
		\draw [style=dash] (9.center) to (10.center);
		\draw [style=simple] (8) to (13);
		\draw [style=simple] (13) to (28.center);
		\draw [style=dash] (4) to (6);
		\draw [style=simple] (38) to (56);
		\draw [style=simple] (56) to (63);
		\draw [style=wide] (63) to (52);
		\draw [style=simple] (52) to (50);
		\draw [style=simple] (38) to (37);
		\draw [style=simple] (37) to (45);
		\draw [style=simple] (45) to (58);
		\draw [style=simple] (58) to (47);
		\draw [style=dash] (43.center) to (41.center);
		\draw [style=dash] (64.center) to (33.center);
		\draw [style=dash] (55.center) to (48.center);
		\draw [style=simple] (50) to (35.center);
		\draw [style=simple] (47) to (65.center);
	\end{pgfonlayer}
\end{tikzpicture}}
\vspace{-3mm}
\caption{ \small Examples of a monotone 10-gon $P$; edge 2-transmitters covering $P$ are thickened. In the proof of Lemma~\ref{lem:monotone10gon}, (a) falls under case 1 and  (b) case 2.}
\label{fig:m-edge-cases}
\end{figure}
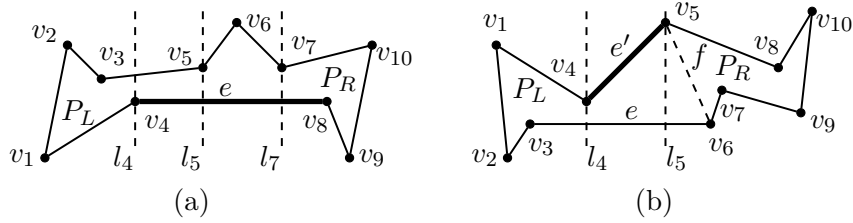

\begin{proof}
 Label vertices of $P$ 
left to right: 
$v_1,\ldots,v_{10}$. 
Assume no two vertices are on the same vertical line (otherwise perturb one by some sufficiently small 
$\varepsilon$).
 
Draw vertical line $l_5$ though 
$v_5$. 
Let $e$ be the edge it hits on the other monotone chain;
its endpoints are $a = v_i$ where $i \leq 4$ and $b = v_j$ where $j \geq 6$. 
First consider the vertical line $l_4$ through $v_4$. This separates from $P$ a 5-gon $P_L \subset P$ with vertices $v_1$, $v_2$, $v_3$, $v_4$, and the other intersection of $l_4$ with $P$'s 
boundary. By Lemma~\ref{lem:5gon}, 
$P_L$ can be covered by a point 2-transmitter placed anywhere in the interior or boundary of 
$P_L$, and in particular 
anywhere along the intersection of $l_4$ with $P$.

{\bf Case 1:} See Figure~\ref{fig:m-edge-cases}(a). If $b = v_j$ and $j \geq 7$, 
the vertical line $l_7$ through $v_7$ cuts a 5-gon $P_R$ from 
$P$, with vertices $v_7$, $v_8$, $v_9$, $v_{10}$, and the 
other intersection of $l_7$ with $P$'s 
boundary.  
By Lemma~\ref{lem:5gon}, $P_R$ can be covered by a 2-transmitter placed anywhere in $P_R$. 
Then $P_L$ is 2-visible from 
$e\cap l_4$, $P_R$ is 2-visible from 
$e\cap l_7$ and the interior of $e$ covers 
$P\setminus (P_L \cup P_R)$. 


{\bf Case 2:}  Otherwise, $b = v_6$; see Figure~\ref{fig:m-edge-cases}(b). Consider the line segment $f$ connecting $v_5$ and $v_6$. 
We have $f\subseteq P$ and it crosses from one monotone chain to the other.
To its right $f$ separates from $P$ a 6-gon  $P_R$.
By Lemma~\ref{lem:monotone6gon}, $P_R$ is 2-visible from a point 2-transmitter placed at $c$, 
$c = v_5$ or $c = v_6$. 
$P$'s edge $e'$ with right endpoint $c$ (possibly $e' = e$) has left endpoint in the set $\{v_1, v_2, v_3, v_4\}$. 
Edge $e'$ covers $P$: $P_R$ is 2-visible from $c$, $P_L$ is 2-visible from 
$e'\cap l_4$, 
and the interior of $e'$ covers $P\setminus (P_L \cup P_R)$. 
\end{proof}

\begin{figure}[t!]
\hspace{-2 in}
  \centering
\begin{minipage}{0.3 in}
 \subfloat[]{
 \includegraphics[width = 2.3in]{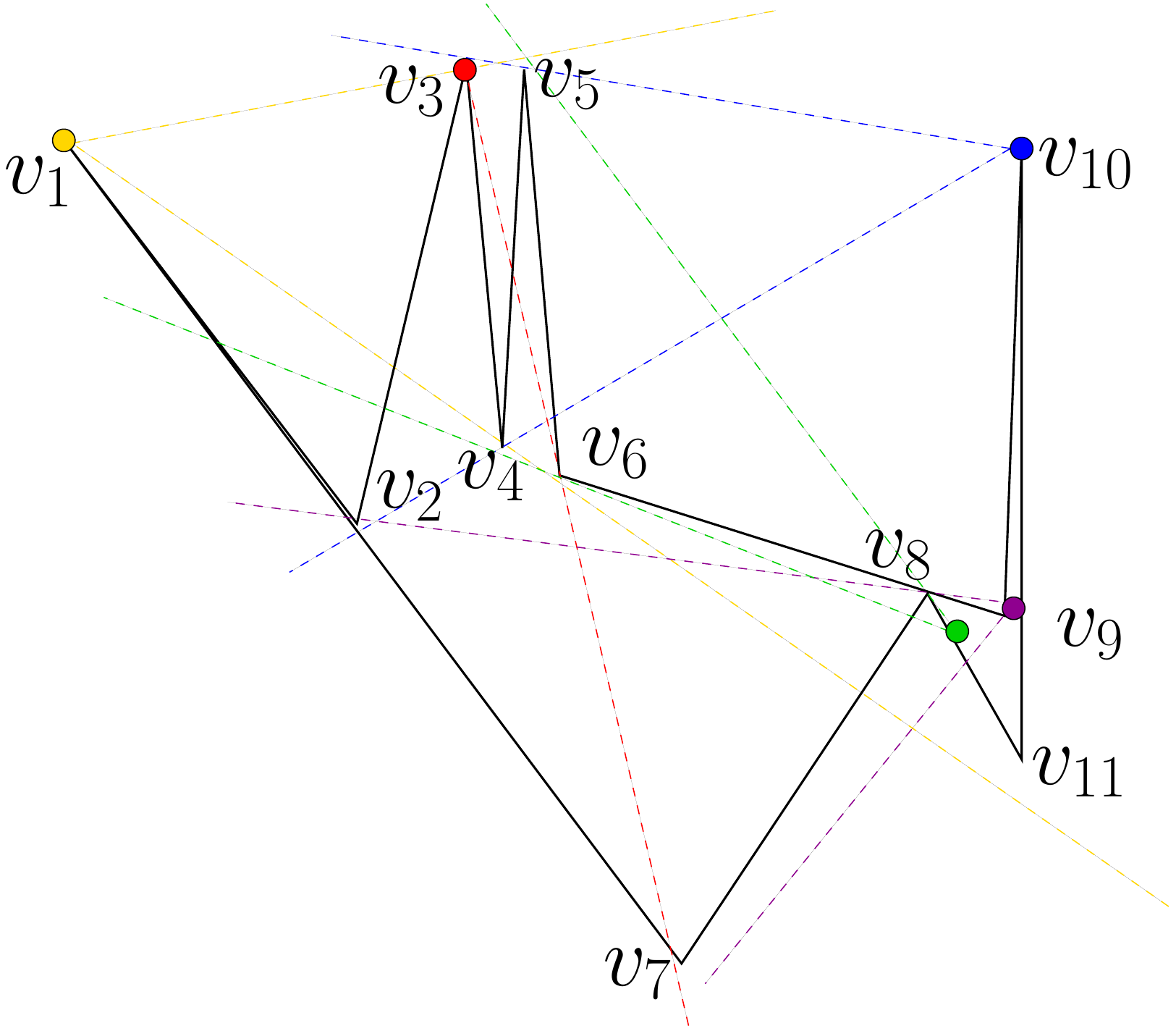}}
 \end{minipage}
  \hspace{2.5 in}
  \begin{minipage}{0.3 in}
  \subfloat[]{
  {\footnotesize
  \begin{tabular}{p{.75cm}|p{0.2cm}p{0.2cm}p{0.2cm}p{0.2cm}p{0.2cm}}
    Edge & \multicolumn{5}{ c }{Witness points} \\
    \midrule
    $v_1 v_7$ & \mycirc[goldenyellow] &  \mycirc[red] & \mycirc[limegreen] && \mycirc[blue]\\
    $v_1 v_2$ & \mycirc[goldenyellow] &  \mycirc[red] & \mycirc[limegreen] &&\\
     $v_2 v_3$ & \mycirc[goldenyellow] &  \mycirc[red] & \mycirc[limegreen] &\mycirc[pansypurple]&\\
     $v_3 v_4$ & \mycirc[goldenyellow] &  \mycirc[red] & &\mycirc[pansypurple]&\mycirc[blue] \\
     $v_4 v_5$ & \mycirc[goldenyellow] &  \mycirc[red] & &\mycirc[pansypurple]&\mycirc[blue] \\
     $v_5 v_6$ & &  \mycirc[red] & \mycirc[limegreen] &\mycirc[pansypurple]&\mycirc[blue] \\
     $v_6 v_9$ & &  \mycirc[red] & \mycirc[limegreen] &\mycirc[pansypurple]&\mycirc[blue] \\
     $v_7 v_8$ & \mycirc[goldenyellow] &  & \mycirc[limegreen] &\mycirc[pansypurple]&\mycirc[blue] \\
     $v_8 v_{11}$ & & & \mycirc[limegreen] &\mycirc[pansypurple]&\mycirc[blue] \\
     $v_9 v_{10}$ & &  & \mycirc[limegreen] &\mycirc[pansypurple]&\mycirc[blue] \\
     $v_{10} v_{11}$ & & & \mycirc[limegreen] &\mycirc[pansypurple]&\mycirc[blue] \\
  \end{tabular}}
  }
  \end{minipage}
 \caption{\small  (a) A monotone 11-gon requiring two edge 2-transmitters; boundaries of 2-visibility regions of the five colored witness points are dotted. (b) $\mycirc[goldenyellow] $ denotes the edge sees the gold witness point, etc. No edge 2-transmitter covers all five witness points.}
\label{fig:m-edge-11gon}
\end{figure}

The bound of 10 
is the best we 
can hope for: 
the monotone 11-gon from Figure~\ref{fig:m-edge-11gon}(a) necessitates two edge 2-transmitters.

The next lemma follows immediately from the proof of the Splitting Lemma in \cite{affhhuv-mimp-09}, and is crucial to the subsequent sufficiency result. 
\begin{lemma}\label{lem:split-edge}
Let $P$, $L$, $R$, $L'$, $R'$, and $l$ be as in the Splitting Lemma. 
Then for every edge $e\neq l$ of $L$, the subset of $e$ left of $l$ is a subset of an edge of $P$. For every edge $e \neq l$ of $R$, the subset of $e$ right of $l$ is a subset of an edge of $P$. 
\end{lemma}
\begin{theorem}\label{thm:m-edge-upper}
$\left\lceil \frac{n-2}{8}\right\rceil$ edge 2-transmitters are always sufficient to cover a monotone $n$-gon with $n \geq 4$.
\end{theorem}
\begin{proof}
We induct on $n$.
Base case: one edge 2-transmitter $e$ covers 
monotone k-gons, $k = 3,4,...,10$ 
by Lemmata~\ref{lem:5gon}, \ref{lem:monotone6gon}, and \ref{lem:monotone10gon}. 
Each $p \in P$ is 2-visible from some $q\in e$, with $q$ to the right of at least two vertices of $P$.

Suppose $n > 10$, and that for all $n' < n$, every monotone $n'$-gon 
$P'$ 
can be covered by a set 
$C$ of $\left\lceil \frac{n'-2}{8} \right\rceil$ edge 2-transmitters, and 
each $p \in P'$ is 2-visible from some point $q\in e \in C$, with $q$ to the right of at least two vertices of $P'$.
Apply the Splitting Lemma (Lemma~\ref{lem:split}) 
for $m = 10$ to obtain a monotone 10-gon $L$ and a monotone $(n-8)$-gon $R$. Let $l$, $L'$, and $R'$ be as in 
the lemma. Then by Lemma~\ref{lem:monotone10gon}, 10-gon $L$ can be covered by a single edge 2-transmitter $e$ such that every point $p \in L$ is 2-visible from some point $q$ on $e$ that is to the left of at least two vertices of $L$ and to the right of at least two vertices of $L$. 
It follows that $e \neq l$, and the portion of $e$ consisting of all such points $q$ is entirely to the left of $l$, so by Lemma~\ref{lem:split-edge} is a subset of an edge of $P$. Thus, there exists a single edge 2-transmitter $t$ of $P$ that covers $L'$. 
Moreover, all points in $L'$ are 2-visible from a point 
$q\in t\in P$, where $q$ is to the right of at least two vertices of $P$, because the same statement holds for edge $e$ of $L$.

By the induction hypothesis, monotone $(n-8)$-gon $R$ can be covered by a collection $C$ of $\left\lceil \frac{(n-8)-2}{8}\right\rceil$ edge 2-transmitters, and every point in $R'$ is 2-visible from some point 
$q\in e \in C$, where $q$ is to the right of at least two vertices of $R$. 
Hence, for each edge $e \in C$, the portion $t$ of $e$ consisting of all such points 
s entirely to the right of $l$, so by Lemma~\ref{lem:split-edge} 
$t\subset$ edge $\in P$. 
Thus, there exists a collection $C'$ of $\left\lceil \frac{(n-8)-2}{8} \right\rceil$ edge 2-transmitters covering $R'$.

So,  $P$ has an edge 2-transmitter cover 
$C$ 
of 
size $
1 + \left\lceil \frac{(n-8)-2}{8} \right\rceil = \left\lceil \frac{n-2}{8}\right\rceil$ with the assumed property.
\end{proof}


\subsection{Monotone Orthogonal Polygons}

We now consider a more restrictive class of simple polygons, monotone orthogonal polygons, and give a tight bound that $\left\lceil \frac{n-2}{10}\right\rceil$ edge 2-transmitters are always sufficient and sometimes necessary to cover monotone polygons with~$n$ vertices. 
\begin{theorem}\label{thm:mo-edge-lower}
There exist monotone orthogonal (MO) $n$-gons that require $\left\lceil \frac{n-2}{10}\right\rceil$ edge 2-transmitters.
\end{theorem}
\begin{proof}
Consider the staircase-shaped polygon of Figure \ref{fig:mo-edge-lower}. 
A single edge 2-transmitter can see the entirety of at most five segments, horizontal or vertical, of the staircase.  The total number of edges is twice the number $k$ of segments plus 2, and $\left\lceil\frac{k}{5} \right\rceil =  \left\lceil \frac{n-2}{10}\right\rceil$ edge 2-transmitters are sometimes necessary.  
\end{proof}

\begin{figure}[t!]

\begin{centering}
\vspace{-25mm}
\includegraphics[scale = 0.5, angle = 45]{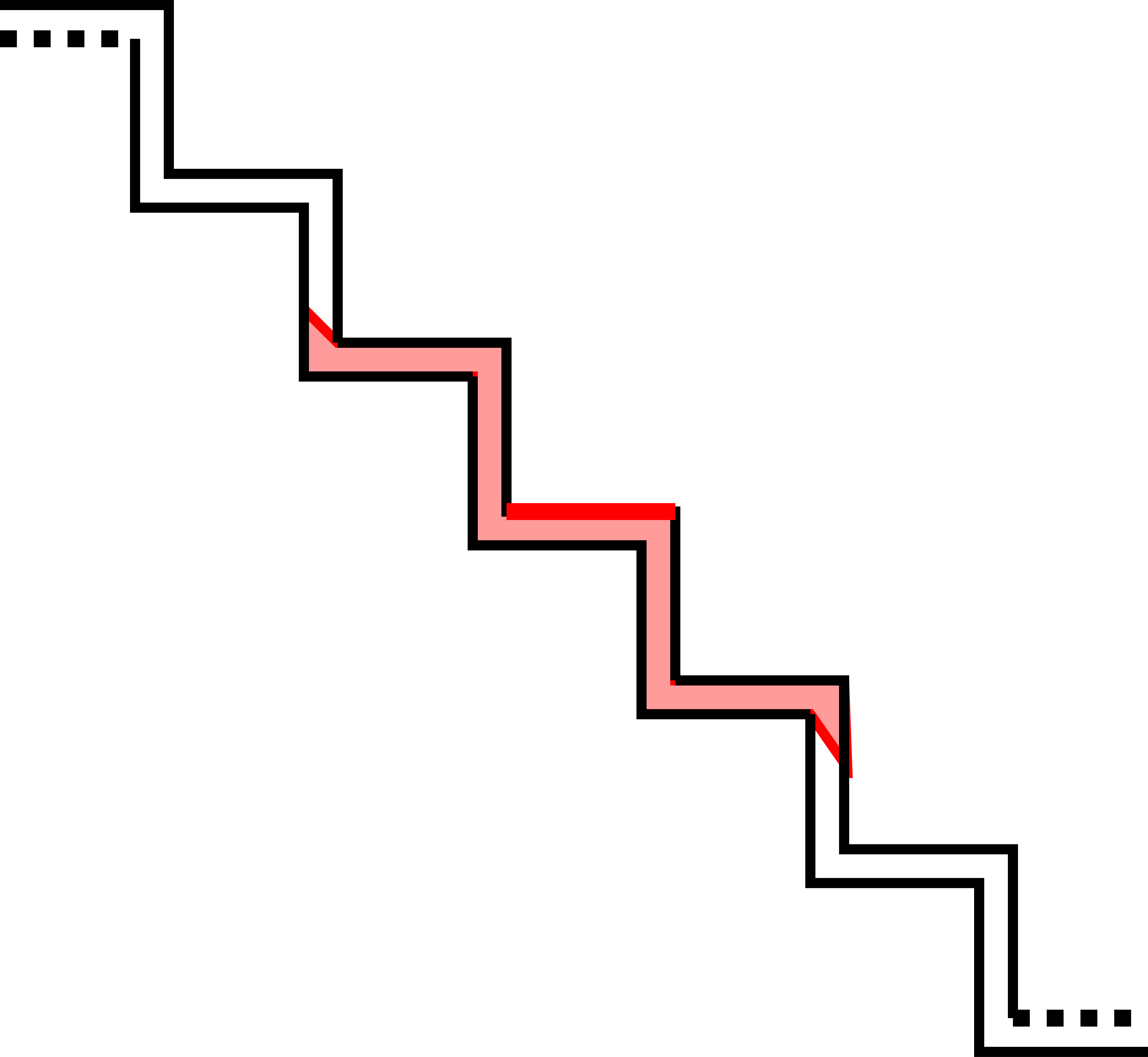}
\vspace{-25mm}

\end{centering}

\caption{\small Lower bound construction for monotone orthogonal polygons, necessitating $\left\lceil \frac{n-2}{10}\right\rceil$ edge 2-transmitters. The 2-visibility region of the bold edge is shaded.}
\label{fig:mo-edge-lower}
\end{figure}

\noindent We now proceed to show the upper bound. 
\begin{lemma}\label{lem:mo6gon}
Any monotone orthogonal 6-gon is covered by a (point) 2-transmitter placed anywhere. 
\end{lemma}
\begin{proof}
Any orthogonal 6-gon $P$ has at most one reflex vertex.  Thus, for any point $p$, the line $\overline{pq}$ for any other point $q\in P$ crosses at most two edges of $P$.
\end{proof}

\begin{lemma}\label{lem:mo12gon}
Any monotone orthogonal 12-gon $P$ can be covered by one edge 2-transmitter, not placed on its leftmost or rightmost edge.
\end{lemma}

\begin{proof}
Assume no two vertical edges in $P$ have the same 
$x$-coordinate. Order the vertices of $P$ 
left to right, 
and order vertices with the same $x$-value (endpoints of a vertical edge $e_v$) 
 according to a left to right traversal of the monotone chain containing $e_v$.
The two leftmost and rightmost vertices can be put in any order.
Label the vertices as $v_1$, $v_2$,..., $v_{12}$ in this order, see Figure~\ref{fig:mo12-gon}. 
Vertical edges have endpoints $v_{2i-1}, v_{2_i}$, 
$i = 1,2,...,6$. 
For horizontal edges (except for rightmost and left most vertices), the right (left) vertex has odd (even) index.

Consider the supporting line $l$ through the vertical edge 
$\overline{v_5,v_6}$.
This separates from 
$P$ a 6-gon with vertices $v_1$, $v_2$, $v_3$, $v_4$, $v_5$, and 
$l\cap e$, where $e$ is an edge in the opposite monotone chain from $v_5$. 
We have $v_6 \subset l$,
though it may lie on the boundary of $P_L$, as in Figure~\ref{fig:mo12-gon}(a), or not, as in Figure~\ref{fig:mo12-gon}(b). 
The supporting line $l'$ through the vertical edge with 
$\overline{v_7,v_8}$
also separates from $P$ a 6-gon $P_R$ with vertices $v_8$, $v_9$, $v_{10}$, $v_{11}$, $v_{12}$ and 
$l'\cap e'$, where $e'$ is an edge in the opposite chain from $v_8$. 
Note edge $e$ always extends rightward past the right boundary of $P_L$, 
thus, its right endpoint is $v_i$ for $i \geq 7$. We consider two cases for the number of edges extending from $P_L$ rightward beyond $l$ (based on $v_6$'s location).

{\bf Case 1:} 
  If $v_6$ is on the boundary of $P_L$, a second edge extends rightward from $v_6$, also with its left endpoint in $P_L$ and its right endpoint some $v_i$ for $i \geq 7$; see Figure~\ref{fig:mo12-gon}(a). In this case at least one of these two edges has right endpoint $v_i$ for $i \geq 8$, 
so it covers $P_R$, $P_L$ and the rectangular region $P \setminus (P_L \cup P_R)$.
  
  {\bf Case 2:} Else, $v_6$ is not on the boundary of $P_L$, and $e$ is the only edge extending rightward from $P_L$; see Figure~\ref{fig:mo12-gon}(b).  
Similar to case 1, if $e$ has right endpoint $v_i$ for $i \geq 8$ we are done, so, we suppose $v_7$ is $e$'s right endpoint.
In this case, $e$ might not cover the entirety of $P_R$. Let $f$ be the segment connecting $v_5$ and $v_7$. 
We have $f\subset P$. We note that 
$f$ separates from 
$P$ a (non-orthogonal) 6-gon $P_L'$ with vertices $v_1, v_2, v_3, v_4, v_5, v_7$, and $P_L \subseteq P_L'$. 
$P_L'$ 
contains at most one reflex vertex and, thus, can be covered by a point 2-transmitter placed anywhere on its boundary.  In particular, the edge $v_7v_8$ covers $P_L'$ from $v_7$, $P_R$ from $v_8$, and $P\setminus (P_L \cup P_R)$ from its interior.
  
  In neither case above do we pick the rightmost or leftmost edge to cover $P$.\end{proof}

\begin{figure}

\begin{centering}

\begin{tikzpicture}[scale=0.6]
	\begin{pgfonlayer}{nodelayer}
		\node [style=vertex] (0) at (3, 1.25) {};
		\node [style=vertex] (1) at (3, -1) {};
		\node [style=vertex] (2) at (7, -1) {};
		\node [style=vertex] (3) at (7, 0.25) {};
		\node [style=vertex] (4) at (8.25, 2) {};
		\node [style=vertex] (5) at (3.75, 1.25) {};
		\node [style=vertex] (6) at (3.75, -0) {};
		\node [style=vertex] (7) at (9.25, 0.25) {};
		\node [style=vertex] (8) at (5.75, -0) {};
		\node [style=vertex] (9) at (5.75, 1) {};
		\node [style=vertex] (10) at (8.25, 1) {};
		\node [style=vertex] (11) at (9.25, 2) {};
		\node [style=none] (12) at (5.75, -2) {};
		\node [style=none] (13) at (5.75, 2.25) {};
		\node [style=none] (14) at (7, 2.25) {};
		\node [style=none] (15) at (7, -2) {};
		\node [style=none] (16) at (5.5, -1.75) {$l$};
		\node [style=none] (17) at (6.75, -1.75) {$l'$};
		\node [style=none] (18) at (3.5, -0.5) {$P_L$};
		\node [style=none] (19) at (8.75, 0.75) {$P_R$};
		\node [style=vertex] (20) at (-0.75, 1.5) {};
		\node [style=none] (21) at (-4.5, -1.75) {$l$};
		\node [style=vertex] (22) at (-6.25, 1.25) {};
		\node [style=vertex] (23) at (-0.75, -0.25) {};
		\node [style=vertex] (24) at (-5.5, -1) {};
		\node [style=vertex] (25) at (-1.75, 0.25) {};
		\node [style=none] (26) at (-3, 2) {};
		\node [style=none] (27) at (-3.25, -1.75) {$l'$};
		\node [style=vertex] (28) at (-3, -1) {};
		\node [style=vertex] (29) at (-4.25, 1.25) {};
		\node [style=vertex] (30) at (-6.25, -0) {};
		\node [style=none] (31) at (-5, 0.5) {$P_L$};
		\node [style=vertex] (32) at (-1.75, 1.5) {};
		\node [style=none] (33) at (-1.25, 0.25) {$P_R$};
		\node [style=none] (34) at (-4.25, -2) {};
		\node [style=vertex] (35) at (-5.5, -0) {};
		\node [style=none] (36) at (-3, -2) {};
		\node [style=none] (37) at (-4.25, 2) {};
		\node [style=vertex] (38) at (-4.25, 0.25) {};
		\node [style=vertex] (39) at (-3, -0.25) {};
		\node [style=none] (40) at (-6.75, 1.5) {$v_1$};
		\node [style=none] (41) at (-6.75, -0) {$v_2$};
		\node [style=none] (42) at (-6, -0.5) {$v_3$};
		\node [style=none] (43) at (-5.75, -1.25) {$v_4$};
		\node [style=none] (44) at (-3.75, 1.5) {$v_5$};
		\node [style=none] (45) at (-3.75, 0.75) {$v_6$};
		\node [style=none] (46) at (-2.5, -1) {$v_7$};
		\node [style=none] (47) at (-2.5, -0.5) {$v_8$};
		\node [style=none] (48) at (-2.25, 0.75) {$v_9$};
		\node [style=none] (49) at (-2.25, 1.75) {$v_{10}$};
		\node [style=none] (50) at (-0.25, 1.75) {$v_{11}$};
		\node [style=none] (51) at (-0.25, -0.75) {$v_{12}$};
		\node [style=none] (52) at (4.25, 1.5) {$v_3$};
		\node [style=none] (53) at (7.5, -0.25) {$v_8$};
		\node [style=none] (54) at (2.5, -1) {$v_2$};
		\node [style=none] (55) at (5.25, 1.25) {$v_6$};
		\node [style=none] (56) at (7.5, -1.25) {$v_7$};
		\node [style=none] (57) at (9.75, 2.25) {$v_{11}$};
		\node [style=none] (58) at (4.25, 0.5) {$v_4$};
		\node [style=none] (59) at (7.75, 2.25) {$v_{10}$};
		\node [style=none] (60) at (2.5, 1.25) {$v_1$};
		\node [style=none] (61) at (5.25, 0.5) {$v_5$};
		\node [style=none] (62) at (9.75, -0.25) {$v_{12}$};
		\node [style=none] (63) at (7.75, 1.5) {$v_9$};
		\node [style=none] (64) at (-4.75, -0.75) {$e$};
		\node [style=none] (65) at (4.75, -0.75) {$e$};
		\node [style=none] (66) at (6.5, -0.25) {$f$};
		\node [style=none] (67) at (-3.75, -2.75) {(a)};
		\node [style=none] (68) at (6.25, -2.75) {(b)};
	\end{pgfonlayer}
	\begin{pgfonlayer}{edgelayer}
		\draw [style=simple] (0) to (1);
		\draw [style=simple] (1) to (2);
		\draw [style=wide] (2) to (3);
		\draw [style=simple] (3) to (7);
		\draw [style=simple] (0) to (5);
		\draw [style=simple] (5) to (6);
		\draw [style=simple] (6) to (8);
		\draw [style=simple] (8) to (9);
		\draw [style=simple] (9) to (10);
		\draw [style=simple] (7) to (11);
		\draw [style=simple] (11) to (4);
		\draw [style=simple] (4) to (10);
		\draw [style=dash] (12.center) to (13.center);
		\draw [style=dash] (14.center) to (15.center);
		\draw [style=dash] (8) to (2);
		\draw [style=simple] (22) to (30);
		\draw [style=simple] (30) to (35);
		\draw [style=simple] (35) to (24);
		\draw [style=simple] (24) to (28);
		\draw [style=simple, in=180, out=0, looseness=1.00] (22) to (29);
		\draw [style=simple] (29) to (38);
		\draw [style=wide] (38) to (25);
		\draw [style=simple] (25) to (32);
		\draw [style=simple] (32) to (20);
		\draw [style=simple] (28) to (39);
		\draw [style=simple] (39) to (23);
		\draw [style=simple] (23) to (20);
		\draw [style=dash] (37.center) to (34.center);
		\draw [style=dash] (26.center) to (36.center);
	\end{pgfonlayer}
\end{tikzpicture}

\end{centering}
\caption{\small Examples of a monotone orthogonal 12-gon $P$; edge 2-transmitters covering $P$ are thickened. In the proof of Lemma~\ref{lem:mo12gon}, (a) falls under case 1 and (b) case 2.}
\label{fig:mo12-gon}
\end{figure}
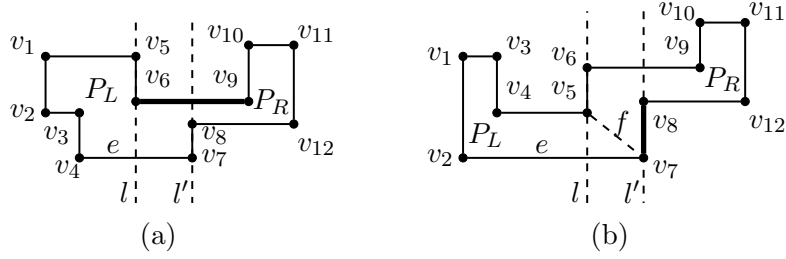

\begin{theorem} \label{thm:mo-edge-upper}
$\left\lceil \frac{n-2}{10} \right\rceil$ edge 2-transmitters are always sufficient to cover a monotone orthogonal $n$-gon.
\end{theorem}
\begin{proof}
We induct on $n$.
For a 
monotone orthogonal $n$-gon $P$,
$n$ is even and $P$ has $n/2$ vertical and horizontal edges.
By  Lemma~\ref{lem:mo12gon}, all 
monotone orthogonal $m$-gons 
with $4\leq m \leq 12$ can be covered by one 
edge 2-transmitter, not placed on its leftmost edge. 

Label the vertical edges of $P$ in order left to right as $e_1, e_2, \,\ldots, e_{n/2}$. Consider the supporting line of edge $e_6$; this separates from $P$ a monotone orthogonal 12-gon $Q$ with six vertical edges $e_1, e_2, \,\ldots, e_5$, and some segment of the supporting line of $e_6$. Polygon $Q$ can be covered by a single edge 2-transmitter, not placed on its leftmost or rightmost edge. The remainder $P \setminus Q$ has $n/2 - 5\geq 2 $ vertical edges $e_7, e_8, \;\ldots, e_{n/2}$, and some segment of the supporting line of $e_6$, 
so, $|V(P \setminus Q)| = n-10$. By the inductive hypothesis, it can be covered by  $\left\lceil \frac{(n-10)-2}{10}\right\rceil$ edge 2-transmitters, none of which are placed on its leftmost edge. Together with the single edge covering $Q$, this yields a cover of $P$ by $
1+\left\lceil \frac{(n-10)-2}{10}\right\rceil= \left\lceil \frac{n-2}{10} \right\rceil$ edge 2-transmitters, not including $P$'s leftmost edge.
\end{proof}

\section{Conclusion}\label{sec:concl}
We presented NP-hardness, necessity and sufficiency results for point and edge 2-transmitters.
For edge 2-transmitters, only the result on monotone orthogonal polygons is tight. For the other classes, a gap between upper and lower bound remains. In particular, the upper bound in two cases comes from the less powerful 0-transmitters.

In addition, the point 2-transmitter problem in general polygons remains 
a very interesting question: we improved the current lower bound to $\lfloor n/5\rfloor$, but the best upper bound is still given by the upper bound from the general AGP, that is, $\lfloor n/3\rfloor$.

\section*{ Acknowledgments.}
We would like to thank Joseph O'Rourke for sharing the 2-transmitter problem with us during the Mathematics Research Communities Workshop on Discrete and Computational Geometry in 2012 
and the AMS for their financial support of this program. 


{
\bibliographystyle{abbrv}
\bibliography{lit}

}


\end{document}